\documentclass[12 pt,draftcls,onecolumn]{IEEEtran}
\usepackage[pdftex]{graphicx}
\usepackage[cmex10]{amsmath}
\usepackage{amssymb}
\usepackage{float}
\usepackage{subfloat}
\DeclareMathAlphabet{\mathscr}{OT1}{pzc}{m}{it}

 \ifCLASSINFOpdf

\else

\fi

\usepackage{fancyhdr}

\begin{document}

\pagestyle{plain}

\thispagestyle{fancy}

\renewcommand{\headrulewidth}{0pt}
\fancyhf{} This work has been submitted to the IEEE for possible
publication. Copyright may be transferred without notice, after
which this version may no longer be accessible.

\title{Network Code Design for Orthogonal Two-hop Network with Broadcasting Relay: A Joint Source-Channel-Network Coding Approach}
\setcounter{page}{0}

\author{\IEEEauthorblockN{Roghayeh Joda and Farshad~Lahouti}\\
\IEEEauthorblockA{Wireless Multimedia Communications Laboratory\\
School of ECE, College of Engineering, University of Tehran\\
Email: [rjoda, lahouti]@ut.ac.ir, http://wmc.ut.ac.ir}
\thanks{}
}


\maketitle

\pagestyle{plain}

\thispagestyle{fancy}

\renewcommand{\headrulewidth}{0pt}

\fancyhf{} \lfoot{\tiny Preliminary reports on this research have
appeared in IEEE Queen's Biennial Symposium on Communications,
Kingston, ON, Canada, June 2008 and IEEE International Conference on
Wireless Communications and Signal Processing, Nanjing, China,
November 2009.}

\begin{abstract}
This paper addresses network code design for robust transmission of
sources over an orthogonal two-hop wireless network with a
broadcasting relay. The network consists of multiple sources and
destinations in which each destination, benefiting the relay signal,
intends to decode a subset of the sources. Two special instances of
this network are orthogonal broadcast relay channel and the
orthogonal multiple access relay channel. The focus is on complexity
constrained scenarios, e.g., for wireless sensor networks, where
channel coding is practically imperfect. Taking a source-channel and
network coding approach, we design the network code (mapping) at the
relay such that the average reconstruction distortion at the
destinations is minimized. To this end, by decomposing the
distortion into its components, an efficient design algorithm is
proposed. The resulting network code is nonlinear and substantially
outperforms the best performing linear network code. A motivating
formulation of a family of structured nonlinear network codes is
also presented. Numerical results and comparison with linear network
coding at the relay and the corresponding distortion-power bound
demonstrate the effectiveness of the proposed schemes and a
promising research direction.
\end{abstract}
\vspace{-2pt}

\section{Introduction}\label{SI}

 \IEEEPARstart{A}{} wireless sensor network consists of
a large number of small, low cost and power constrained sensor
nodes, which are spatially distributed and communicate through
wireless channels. The nodes power limitation and the wireless
propagation loss indicate that these nodes can communicate over
short distances. Communications over longer ranges may be
facilitated with multiple hops. In other words, the source
signal is relayed by one or more nodes in the network and
forwarded to the receiving nodes. We refer to a two-hop network with
a single broadcasting relay as the TN-BR. In a TN-BR, utilizing the
relay signal, each destination node intends to decode the signals
transmitted by a certain subset of the sources. Two particular TN-BR
scenarios are broadcast relay channel (BRC) (Fig. \ref{F1}(a)) and
multiple access relay channel (MARC) (Fig. \ref{F1}(b)). In such
networks, a dedicated relay assists reliable transmission by
broadcasting an appropriate signal, based on a certain relaying
scheme, to the destination(s) \cite{R1}\cite{R2}.

In this paper, we consider a special case of the TN-BR with $N$
source nodes in which every source communicates with its intended
destination(s) in a distinct orthogonal channel without interference.
The relay is half duplex and listens to whatever the sources
transmit; it then broadcasts to all the destinations in another
orthogonal channel. We refer to such a network, with $N+1$
orthogonal channels, as orthogonal TN-BR (OTN-BR). In this
orthogonal setting, a BRC may be viewed as one with several
source-destination pairs with corresponding orthogonal channels and
a dedicated broadcasting relay, referred to as the orthogonal
multi-user channel with broadcasting relay (OMC-BR)
\cite{R3}\nocite{R4}-\cite{R5}. This paper presents network code
design at the relay for an OTN-BR with constrained resources and
considers orthogonal MARC (OMARC) and OMC-BR as special cases.

Capacity achieving approaches in wireless channels with relay
require substantial memory at the relay to accommodate inter-block
encoding and large block sizes \cite{R6}-\nocite{R30}\cite{R31}.
Though, in the context of complexity constrained wireless sensor
networks simple channel codes and optimized memoryless mappings at
the relay \cite{R7}\nocite{R8,R9}-\cite{R10} are of practical
interest. In \cite{R7} and \cite{R8}, for a single source
communicating a continuous signal to a single destination over a
relay channel or through multi-hop relays, memoryless relay mappings
maximizing the SNR at the destination are presented. In IEEE
802.15.4 \cite{R11}, the widely popular standard for wireless sensor
networks, in fact, no channel coding is considered. This implies
that links are not practically error free and the errors are to be
tackled at higher layers of protocol stack. One option is to devise
an automatic repeat request protocol at the link layer \cite{R12}.
In this work, we take a joint source-channel-network coding approach
to deal with the error at the presentation and network layers.

In a general network, network coding at the relay nodes may be
considered for improved communications performance. In \cite{R13},
it is shown that multicast capacity (maximum multicast rate) in an
error free network can be achieved by network coding. In this
setting, the intermediate nodes are allowed to decode and re-encode
their incoming information. In \cite{R14}, it is demonstrated that
linear network coding with finite alphabet size can achieve the
multicast capacity for communications over a network with error-free
links. This corresponds to a separate design of network and channel
coding, in the sense that perfect point to point channel codes
render the links error-free from a network layer perspective. In
\cite{R15}, it is shown that source-network coding separation or
channel-network coding separation is not optimal for non-multicast
networks. In \cite{R16}, it is shown that source and network coding
separation is not optimal even for general multicast networks and
that it is optimal only for networks with two sources and two
receivers. Effective joint network-channel codes are presented for
robust multicast over networks with noisy links in, e.g.,
\cite{R17}\nocite{R18,R19,R20,R21,R27}-\cite{R28}.

In this paper, taking a joint source-channel-network coding approach
for the OTN-BR, we present optimized network coding schemes at the
relay for efficient reconstruction of the source signals at the
destinations. In particular, we focus on resource constrained
wireless networks, e.g., a wireless sensor network, where it is
assumed that the mapping at the relay is memoryless and the channel
coding is practically imperfect. First, a decoding scheme is
presented for the minimum mean squared error (MMSE) reconstruction
of signals transmitted from sources over the noisy channels of an
OTN-BR at the destinations. As we shall demonstrate, in such a
setting, the average reconstruction distortion at the OTN-BR
destinations is decomposed into two parts referred to as the source
distortion and the network distortion. The former is due to source
coding and the latter primarily corresponds to the network, i.e.,
its channels, decoding and encoding (mapping) at the relay and
decoding at the destinations. The network distortion is also
influenced by the source coder output statistics. The objective is
then to design a proper network code at the relay that minimizes the
network distortion. This in turn lends itself to an efficient design
scheme based on simulated annealing (SA), which results in a
nonlinear network code at the relay. Taking insight of the resulting
mappings, we next consider a motivating formulation of a family of
structured nonlinear network codes with substantially reduced design
complexity. As a benchmark for comparisons, we also present
performance cut-set bounds for OTN-BR.

Numerical results are provided demonstrating the superior
performance of the proposed codes when compared to linear network
coding. The outcomes are particularly inspiring as they demonstrate
(i) the insufficiency of linear network coding in such wireless
networks, (ii) a constructive approach to design nonlinear network
codes and (iii) the effectiveness of the proposed joint
source-channel-network coding schemes in complexity constrained
wireless networks in comparison to the performance bounds.

The rest of this article is organized as follows. Following the
preliminaries and the description of system model in Section
\ref{SII}, Section \ref{SIII} presents a MMSE decoder for OTN-BR.
Next, in Section \ref{SIV} the distortion at the OTN-BR destinations
is analyzed. In Section \ref{SV}, a method to design an optimized
network code at the relay is presented and then a family of
structured nonlinear network codes for the OTN-BR is suggested.
Sections \ref{SVI} and \ref{SVII} present the separate source and
channel-network coding bounds and the performance evaluations and
comparisons, respectively. This article is concluded in Section
\ref{SVIII}.

\section{Preliminaries} \label{SII}
\subsection{Notation}\label{SII-A}
The following notations are used in this paper. Capital letters,
e.g., $I$, represent random variables and small letters, e.g., $i$,
represent the realizations of random variables. We replace the
probability $P(I=i)$ by $P(I)$ for simplicity. The vectors are shown
bold faced, e.g., \textit{\textbf{I}}. The sets are identified by
scripts, e.g., $\mathcal{A}$. For the set
$\mathcal{A}\subset\mathcal{B}$, $\mathcal{A}^{c}$ denotes the
complement of $\mathcal{A}$ in $\mathcal{B}$. We denote the sets of
random variables $\{U_{k}: k\in\mathcal{A}\}$,
$\{\textit{U}_{m,k}:k\in\mathcal{A}\}$ and
$\{\textit{U}_{k,m}:k\in\mathcal{A}\}$ by $U_\mathcal{A}$,
$\textit{U}_{m,\mathcal{A}}$ and $\textit{U}_{\mathcal{A},m}$,
respectively.

\subsection{OTN-BR}\label{SII-B}
In an OTN-BR, there are $N$ sources, $M$ destinations and one relay.
The sets $\mathcal{S}=\{1,...,N\}$ and $\mathcal{D}=\{1,...,M\}$,
respectively indicate the set of sources and destinations. The
transmitted signal of source $s\in\mathcal{S}$ is to be decoded at
each of the nodes in $\mathcal{D}_s$. The subset of the sources that
are to be decoded at the destination $d\in\mathcal{D}$ are denoted
by $\mathcal{S}_d$.

In Section \ref{SI}, we presented OMARC and OMC-BR as particular
instances of OTN-BR. In addition to these two networks, we also
consider another instance of OTN-BR with two sources and three
destinations, denoted by OTN-BR-(2,3) (Fig. \ref{F1}(c)). While the developments in subsequent sections pertains to arbitrary OTN-BRs, we exemplify and assess the performance of the proposed code
construction and decoding algorithms based on these three OTN-BR
networks. In an OTN-BR-(2,3), the destination $d$, $d\in\{1,2\}$ is
to decode its corresponding source $s=d$, while the destination
$d=3$ is to decode the signals transmitted from both sources. Thus,
in this network, $M=3$, $N=2$, $\mathcal{S}_{d=1}=\{1\}$,
$\mathcal{S}_{d=2}=\{2\}$, $\mathcal{S}_{d=3}=\{1,2\}$,
$\mathcal{D}_{s=1}=\{1,3\}$ and $\mathcal{D}_{s=2}=\{2,3\}$.
Furthermore, for an OMC-BR, $M=N$,
$\mathcal{S}_d=\{d\}:d\in\mathcal{D}$ and
$\mathcal{D}_s=\{s\}:s\in\mathcal{S}$, while for an OMARC, $M=1$,
$\mathcal{S}_{d=1}=\mathcal{S}$ and
$\mathcal{D}_s=\{1\}:s\in\mathcal{S}$.

\subsection{System Settings}\label{SII-C}
At the source node $s$, $s\in\mathcal{S}$ the source signal $X_s$ is
scalar quantized with rate $R_{s}$, mapped to a codeword (index) and
transmitted to the destination nodes $\mathcal{D}_s$. The number of
quantization partitions (Voronoi regions) for source node $s$ is
equal to $L_s=2^{R_{s}}$ and we refer to the partition $I_s$  of
this quantizer as $\textit{V}_{I_s}$, where $I_s\in\mathcal{I}_s$,
$\mathcal{I}_s=\{1,\dots,L_s\}$. If the signal $X_s$ belongs to
$\textit{V}_{I_s}$, the corresponding $R_{s}$ bit index $I_s$ is
transmitted over a memoryless noisy channel orthogonal to those of
other users. This is accomplished, e.g., based on a time
multiplexing scheme \cite{R22} and BPSK modulation. In this case,
when the source node $s$, $s\in\mathcal{S}$ transmits, the relay and
destination nodes are in the receive mode; subsequently, the relay
transmits and all the destination nodes are in the receive mode.
Thus, the transmitted index $I_s$, $s\in\mathcal{S}$ is received over noisy channels at
the relay and also at the corresponding destinations $\mathcal{D}_s$
as vectors $\textit{\textbf{Y}}_{r,s}$ and
$\textit{\textbf{Y}}_{\mathcal{D}_s,s}=\{\textit{\textbf{Y}}_{d,s}:d\in\mathcal{D}_s\}$
each with $R_{s}$ components, respectively.

As elaborated in Section \ref{SI}, focusing on a complexity
constrained solution; we assume memoryless mappings at the relay.
Therefore, the relay decodes the received vectors
$\textit{\textbf{Y}}_{r,s}$, $s\in\mathcal{S}$ to the indexes
$\hat{I}_s$, $s\in\mathcal{S}$ based on a maximum a posteriori (MAP)
symbol decoding rule, i.e.,
$\hat{I}_s=\underset{I_s\in\mathcal{I}_s}{\operatorname{argmax}}\;{P\left(I_s|\textit{\textbf{Y}}_{r,s}\right)}$\footnote{Note
that $|\mathcal{I}_s|=2^{R_{s}}$ and the quantizer bit-rate $R_{s}$
is limited, hence the decoding complexity is fairly small.}. The
relay then broadcasts the index $I_r=f(\hat{I}_\mathcal{S})$ with
$R_r$ bits to all destinations, where $I_r\in\mathcal{I}_r$,
$\mathcal{I}_r=\{1,\dots,L_r=2^{R_r}\}$. Note that
$f(\hat{I}_\mathcal{S})$ denotes an arbitrary (not necessarily
linear) network coding function at the relay. The transmitted relay
index $I_r$ is received at the destination node $d$,
$d\in\mathcal{D}$ as $\textit{\textbf{Y}}_{d,r}$ with $R_r$
components. Thus, the rate of the network code is defined as
$R_{NC}\overset{\triangle}{=}{R_r}/{\sum_{s\in\mathcal{S}}{R_s}}$.

Each of the communication channels is assumed memoryless, without
intersymbol interference and described by the transition
probabilities $P(\textit{\textbf{Y}}_{j,k}|I_k)$, where
$(k,j)\in\mathcal{G}$ and
\begin{equation}\label{E1}
\mathcal{G}=\{(k,j):k\in\mathcal{S},j\in\mathcal{D}_k\}\cup\{(k,j):k=r,j\in\mathcal{D}\}\cup\{(k,j):k\in\mathcal{S},j=r\}\textrm{.}
\end{equation}
This may be motivated by the use of an interleaver or a scrambler at
the physical layer. In this paper, we aim at designing the network
code or mapping at the relay such that the average reconstruction
distortion at the destinations is minimized.

\subsection{Definitions}
The developments in subsequent sections in general concern OTN-BR in
arbitrary settings. However, the following particular scenarios are
also considered in parts.

\newtheorem{mydef}{Definition}

\begin{mydef}The ``noiseless relay channels'' corresponds to
the scenario, where in an OTN-BR the channels from sources to the
relay and also from the relay to the destinations are (almost)
noiseless (high SNR). In this case, the source-destination channels
can be noisy.
\end{mydef}

\begin{mydef}The ``noiseless relay channels and very noisy
source-destination channels'' corresponds to the scenario, where in
an OTN-BR the source-destination channels are very noisy and the
channels from sources to the relay and also from the relay to the
destinations are noiseless. When a channel is very noisy (low-end of SNR range), the output of the channel is (almost) independent of its input.
\end{mydef}

\section{MMSE Decoding at OTN-BR Destinations} \label{SIII}
At each destination, the objective is to produce the minimum mean
squared error estimation of the signals transmitted from its
corresponding sources exploiting their dependencies with the signal
received from the relay. According to the described system model, we
take a joint source-channel-network coding approach that includes
the effect of quantization as source coding. In line with the joint
source-channel coding literature, the mean
squared error (MSE) is a desired performance
(distortion) criterion \cite{R29}.
\newtheorem{proposition}{\textit{Proposition}}
\begin{proposition}\label{P1}
In an OTN-BR destination node $d\in\mathcal{D}$, given
$\textit{\textbf{Y}}_{d,\mathcal{S}_d}$ and
$\textit{\textbf{Y}}_{d,r}$, respectively the received signals from
sources $s\in\mathcal{S}_d$ and relay, the optimum MMSE
reconstruction (estimation) of the transmitted source signal $s$ is
given by:

\begin{equation}\label{E2}
\begin{split}
\hat{x}_{s,d}=g_{s}(\textit{\textbf{Y}}_{d,\mathcal{S}_d},\textit{\textbf{Y}}_{d,r})&=E\left[X_s|\textit{\textbf{Y}}_{d,\mathcal{S}_d},\textit{\textbf{Y}}_{d,r}\right]\\
         &=\frac{1}{K}\sum_{I_r}\sum_{I_{\mathcal{S}_d}}E\left[X_s|I_s\right]P(\textit{\textbf{Y}}_{d,\mathcal{S}_d}|I_{\mathcal{S}_d})P(\textit{\textbf{Y}}_{d,r}|I_r)P(I_r|I_{\mathcal{S}_d})P(I_{\mathcal{S}_d})\textrm{,}
\end{split}
\end{equation}
where
\begin{equation}\label{E3}
K=\sum_{I_r}\sum_{I_{\mathcal{S}_d}}P(\textit{\textbf{Y}}_{d,\mathcal{S}_d}|I_{\mathcal{S}_d})P(\textit{\textbf{Y}}_{d,r}|I_r)P(I_r|I_{\mathcal{S}_d})P(I_{\mathcal{S}_d})
\end{equation}
is a factor which normalizes the sum of probabilities to one.
\end{proposition}

\begin{proof}
In \eqref{E2}, $C_{I_s}\overset{\triangle}{=}{E[X_s|I_s]}$,
$s\in\mathcal{S}_d$, describes the codebook at the destination $d$
corresponding to the source $s\in\mathcal{S}_d$. We have
$P(\textit{\textbf{Y}}_{d,\mathcal{S}_d}|I_{\mathcal{S}_d})=\underset{s\in\mathcal{S}_d}\prod{P(\textit{\textbf{Y}}_{d,s}|I_s)}$,
and the terms $P(\textit{\textbf{Y}}_{d,s}|I_s)$, $s\in
\mathcal{S}_d$ and $P(\textit{\textbf{Y}}_{d,r}|I_r)$ represent the
channel transition probability for the channels from source $s$,
and relay to destination $d$, respectively. The term $P(I_r|I_{\mathcal{S}_d})$ in \eqref{E2} and
\eqref{E3} is given by
\begin{equation} \label{E4}
P(I_r|I_{\mathcal{S}_d})=\sum_{\hat{I}_\mathcal{S}}{P(I_r|\hat{I}_{\mathcal{S}})P(\hat{I}_{\mathcal{S}_d}|I_{\mathcal{S}_d})}\prod_{j\in\mathcal{S}_d^c}{P(\hat{I}_j)}\textrm{,}
\end{equation}
in which $\hat{I}_{\mathcal{S}_d}=\{\hat{I}_s:s\in\mathcal{S}_d\}$
and $\hat{I}_s$ is the decoded index (symbol) at the relay
corresponding to the index $I_s$ emitted from the source $s$. The
term $P(I_r|\hat{I}_\mathcal{S})\in\{0,1\}$ corresponds to the
mapping or the network code function $I_r=f(\hat{I}_\mathcal{S})$ at
the relay. In \eqref{E4}, we have
$P(\hat{I}_s)=\underset{I_s}\sum{P(\hat{I}_s|I_s)P(I_s)}$, where
$P(\hat{I}_s|I_s)$, $s\in\mathcal{S}$ indicates the transition
probability of the equivalent discrete source-relay channel. The
proof is completed in Appendix \ref{A1}.
\end{proof}

The value of $C_{I_s}$ and $P(I_{s})$, $s\in\mathcal{S}$ are
acquired from the source coder and stored at every destination $d$,
$d\in\mathcal{D}_s$. As the quantizers bit-rates are relatively
small, this incurs a limited complexity. Given that the channel
transition probabilities
$P(\textit{\textbf{Y}}_{d,\mathcal{S}_d}|I_{\mathcal{S}_d})$ and
$P(\textit{\textbf{Y}}_{d,r}|I_r)$, and the network code
$f(\hat{I}_\mathcal{S})$ are available at the destination $d$, the
RHS of (2) can be effectively computed.

The following two corollaries present Proposition \ref{P1} for
OTN-BR in scenarios described by Definitions 1 and 2. These may be
utilized to invoke simplified network code design procedures, as
elaborated in Sections \ref{SV} and \ref{SVII}.

\newtheorem{corollary}{{\textit{Corollary}}}
\begin{corollary}\label{C1}
In an OTN-BR with ``noiseless relay channels'', given
$\textit{\textbf{Y}}_{d,\mathcal{S}_d}$ and $I_r$ at the destination
$d\in\mathcal{D}$, the MMSE estimation of the transmitted signal
from the source $s\in\mathcal{S}_d$ is given by
\begin{equation}\label{E5}
\hat{x}_{s,d}=g_{s}^1(\textit{\textbf{Y}}_{d,\mathcal{S}_d},I_r)\overset{\triangle}=E\left[X_s|\textit{\textbf{Y}}_{d,\mathcal{S}_d},I_r\right]=\frac{1}{K}\sum_{I_\mathcal{S},f(I_\mathcal{S})=I_r}E\left[X_s|I_s\right]P(\textit{\textbf{Y}}_{d,\mathcal{S}_d}|I_{\mathcal{S}_d})P(I_{\mathcal{S}})\textrm{,}
\end{equation}
in which, $K$ is a factor which normalizes the sum of probabilities
to one.
\end{corollary}

\begin{corollary}\label{C2}
In an OTN-BR with ``noiseless relay channels and very noisy
source-destination channels'', given $I_r$ at the destination
$d\in\mathcal{D}$, the MMSE estimation of the transmitted source
signal $s\in\mathcal{S}_d$ is given by
\begin{equation}\label{E6}
\hat{x}_{s,d}=\hat{x}_{s}=g_{s}^2(I_r)\overset{\triangle}=E\left[X_s|I_r\right]=\frac{1}{K}\sum_{I_\mathcal{S},f(I_\mathcal{S})=I_r}E\left[X_s|I_s\right]P(I_{\mathcal{S}})\textrm{,}
\end{equation}
in which, $K$ is a factor which normalizes the sum of probabilities
to one.
\end{corollary}
Using Proposition \ref{P1}, obtaining the results in Corollaries
\ref{C1} and \ref{C2} is straightforward.

\section{Distortion at Destinations}\label{SIV}
In this section, we investigate the average reconstruction
distortion at the destinations in an OTN-BR. Specifically, we focus
on how the distortion may be decomposed into its components to facilitate network code design at the
relay. The average reconstruction distortion (MSE) at
the destinations can be expressed as follows
\begin{equation}\label{E7}
\begin{split}
D&=\frac{1}{\underset{d=1}{\overset{M}\sum}{\left|\mathcal{S}_d\right|}}\, \underset{d=1}{\overset{M}\sum}\, \sum_{s\in\mathcal{S}_d}E\left[X_s-\hat{X}_{s,d}\right]^2\\
&=\frac{1}{\underset{d=1}{\overset{M}\sum}{\left|\mathcal{S}_d\right|}}\underset{d=1}{\overset{M}\sum}\sum_{s\in\mathcal{S}_d}\sum_{I_r}\sum_{I_{\mathcal{S}_d}}
\underset{\textit{V}_{I_{\mathcal{S}_d}}}\int\,
\underset{\textit{\textbf{Y}}_{d,\mathcal{S}_d}}\int\,
\underset{\textit{\textbf{Y}}_{d,r}}\int
\left|{X_s-g_{s}\left(\textit{\textbf{Y}}_{d,\mathcal{S}_d},\textit{\textbf{Y}}_{d,r}\right)}\right|^2\\
&\times{P\left(I_r|I_{\mathcal{S}_d}\right)P\left(\textit{\textbf{Y}}_{d,\mathcal{S}_d}|I_{\mathcal{S}_d}\right)P\left(\textit{\textbf{Y}}_{d,r}|I_r\right)
P\left(X_{\mathcal{S}_d}\right)\mathrm{d}\textit{\textbf{Y}}_{d,\mathcal{S}_d}\mathrm{d}\textit{\textbf{Y}}_{d,r}\mathrm{d}X_{\mathcal{S}_d}}\textrm{,}
\end{split}
\end{equation}
in which
$\hat{X}_{s,d}=g_{s}\left(\textit{\textbf{Y}}_{d,\mathcal{S}_d},\textit{\textbf{Y}}_{d,r}\right)$,
$s\in\mathcal{S}$ is given in Proposition \ref{P1},
$\left|\mathcal{S}_d\right|$ is the number of elements in
$\mathcal{S}_d$, $\textit{V}_{I_{\mathcal{S}_d}}=\left\{
\textit{V}_{I_k}:k\in\mathcal{S}_d\right\}$ denotes (Voronoi
regions) quantization partitions $I_k$ of sources
$k\in\mathcal{S}_d$ and $P\left(I_r|I_{\mathcal{S}_d}\right)$ given
in \eqref{E4} represents the effect of source-relay channels, and
decoding and network coding at the relay.

\begin{proposition}\label{P2}
The average distortion at the OTN-BR is equal to the sum of two
terms as follows $D=D_{sources}+D_{network}$ in which
\begin{equation}\label{E8}
D_{sources}=\frac{1}{\underset{d=1}{\overset{M}\sum}{\left|\mathcal{S}_d\right|}}\,
\underset{d=1}{\overset{M}\sum}\,
\sum_{s\in\mathcal{S}_d}\sum_{I_s}\underset{\textit{V}_{I_s}}\int\left|X_s-C_{I_s}\right|^2P\left(X_s\right)\mathrm{d}X_s
\end{equation}
\vspace{-2pt} and\vspace{-2pt}
\begin{equation}\label{E9}
\begin{split}
D_{network}&=\frac{1}{\underset{d=1}{\overset{M}\sum}{\left|\mathcal{S}_d\right|}}\,
\underset{d=1}{\overset{M}\sum}\,
\sum_{s\in\mathcal{S}_d}\sum_{I_r}\sum_{I_{\mathcal{S}_d}}P\left(I_{\mathcal{S}_d}\right)
\underset{\textit{\textbf{Y}}_{d,\mathcal{S}_d}}\int\,
\underset{\textit{\textbf{Y}}_{d,r}}\int
\left|{C_{I_s}-g_{s}\left(\textit{\textbf{Y}}_{d,\mathcal{S}_d},\textit{\textbf{Y}}_{d,r}\right)}\right|^2\\
&\times{P\left(I_r|I_{\mathcal{S}_d}\right)P\left(\textit{\textbf{Y}}_{d,\mathcal{S}_d}|I_{\mathcal{S}_d}\right)P\left(\textit{\textbf{Y}}_{d,r}|I_r\right)
\mathrm{d}\textit{\textbf{Y}}_{d,\mathcal{S}_d}\mathrm{d}\textit{\textbf{Y}}_{d,r}}\textrm{,}
\end{split}
\end{equation}
where
$g_{s}\left(\textit{\textbf{Y}}_{d,\mathcal{S}_d},\textit{\textbf{Y}}_{d,r}\right)$
and $P(I_r|I_{\mathcal{S}_d})$ are given in Proposition \ref{P1} and
\eqref{E4}, respectively.
\end{proposition}

\begin {proof}
Replacing
$\left|X_s-g_{s}\left(\textit{\textbf{Y}}_{d,\mathcal{S}_d},\textit{\textbf{Y}}_{d,r}\right)\right|$
by
$\left|X_s-C_{I_s}+C_{I_s}-g_{s}\left(\textit{\textbf{Y}}_{d,\mathcal{S}_d},\textit{\textbf{Y}}_{d,r}\right)\right|$
in \eqref{E7}, we can write $D$ as the sum of three terms, i.e.,
$D=D_{sources}+D_0+D_{network}$, in which
\begin{equation}\label{E23a}
\begin{split}
D_{sources}=&\frac{1}{\underset{d=1}{\overset{M}\sum}{\left|\mathcal{S}_d\right|}}\underset{d=1}{\overset{M}\sum}\sum_{s\in\mathcal{S}_d}\sum_{I_r}\sum_{I_{\mathcal{S}_d}}
\underset{\textit{V}_{I_{\mathcal{S}_d}}}\int\,
\underset{\textit{\textbf{Y}}_{d,\mathcal{S}_d}}\int\,
\underset{\textit{\textbf{Y}}_{d,r}}\int \left|X_s-C_{I_s}\right|^2P\left(I_r|I_{\mathcal{S}_d}\right)\\
&\times
P\left(\textit{\textbf{Y}}_{d,\mathcal{S}_d}|I_{\mathcal{S}_d}\right)P\left(\textit{\textbf{Y}}_{d,r}|I_r\right)
P\left(X_{\mathcal{S}_d}\right)\mathrm{d}\textit{\textbf{Y}}_{d,\mathcal{S}_d}\mathrm{d}\textit{\textbf{Y}}_{d,r}\mathrm{d}X_{\mathcal{S}_d}\textrm{,}
\end{split}
\end{equation}
\begin{equation}\label{E23b}
\begin{split}
D_0=&\frac{1}{\underset{d=1}{\overset{M}\sum}{\left|\mathcal{S}_d\right|}}\underset{d=1}{\overset{M}\sum}\sum_{s\in\mathcal{S}_d}\sum_{I_r}\sum_{I_{\mathcal{S}_d}}
\underset{\textit{V}_{I_{\mathcal{S}_d}}}\int\,
\underset{\textit{\textbf{Y}}_{d,\mathcal{S}_d}}\int\,
\underset{\textit{\textbf{Y}}_{d,r}}\int
2(X_s-C_{I_s})\left(C_{I_s}-g_{s}\left(\textit{\textbf{Y}}_{d,\mathcal{S}_d},\textit{\textbf{Y}}_{d,r}\right)
\right)\\
&\times{P\left(I_r|I_{\mathcal{S}_d}\right)P\left(\textit{\textbf{Y}}_{d,\mathcal{S}_d}|I_{\mathcal{S}_d}\right)P\left(\textit{\textbf{Y}}_{d,r}|I_r\right)
P\left(X_{\mathcal{S}_d}\right)\mathrm{d}\textit{\textbf{Y}}_{d,\mathcal{S}_d}\mathrm{d}\textit{\textbf{Y}}_{d,r}\mathrm{d}X_{\mathcal{S}_d}}
\end{split}
\end{equation}
and
\begin{equation}\label{E23c}
\begin{split}
D_{network}=&\frac{1}{\underset{d=1}{\overset{M}\sum}{\left|\mathcal{S}_d\right|}}\underset{d=1}{\overset{M}\sum}\sum_{s\in\mathcal{S}_d}\sum_{I_r}\sum_{I_{\mathcal{S}_d}}
\underset{\textit{V}_{I_{\mathcal{S}_d}}}\int\,
\underset{\textit{\textbf{Y}}_{d,\mathcal{S}_d}}\int\,
\underset{\textit{\textbf{Y}}_{d,r}}\int
\left|C_{I_s}-g_{s}\left(\textit{\textbf{Y}}_{d,\mathcal{S}_d},\textit{\textbf{Y}}_{d,r}\right)\right|^2\\
&\times{P\left(I_r|I_{\mathcal{S}_d}\right)P\left(\textit{\textbf{Y}}_{d,\mathcal{S}_d}|I_{\mathcal{S}_d}\right)P\left(\textit{\textbf{Y}}_{d,r}|I_r\right)
P\left(X_{\mathcal{S}_d}\right)\mathrm{d}\textit{\textbf{Y}}_{d,\mathcal{S}_d}\mathrm{d}\textit{\textbf{Y}}_{d,r}\mathrm{d}X_{\mathcal{S}_d}}\mathrm{.}
\end{split}
\end{equation}
Noting that the sources are independent and
$\underset{I_s}\sum{\underset{\textit{V}_{I_s}}\int
{P(X_s)\mathrm{d}X_s}}=1$, $D_{sources}$ in \eqref{E23a} is
simplified to \eqref{E8}. Considering
$\underset{\textit{V}_{I_s}}\int\left(X_s-C_{I_s}\right)P(X_s)\mathrm{d}X_s=0$,
it is straightforward to see that $D_0=0$. Finally, noting that
$\underset{\textit{V}_{I_{\mathcal{S}_d}}}\int{P(X_{\mathcal{S}_d})\mathrm{d}X_{\mathcal{S}_d}}=P(I_{\mathcal{S}_d})$,
we obtain $D_{network}$ as in \eqref{E9}.
\end{proof}

As expected, the term $D_{sources}$ only depends on the distortion
due to source coding (quantization) and is independent of the
network code at the relay or of the channel. On the other hand, only
the term $D_{network}$ depends on the channels and the network code
at the relay. Hence, we design the network code such that
$D_{network}$ is minimized. As evident, $D_{network}$
also depends on source coder output statistics and MMSE
decoding of the source signals at
destinations. Therefore, the design takes a joint
source-channel-network coding approach.

The following two corollaries present Proposition \ref{P2} for
OTN-BR in scenarios described by Definitions 1 and 2. These may be
utilized to invoke simplified network code design procedures, as
elaborated in Sections \ref{SV} and \ref{SVII}.
\begin{corollary}\label{C3}
In an OTN-BR with ``noiseless relay channels'', $D_{network}$ is
given by:
\begin{equation}\label{E10}
\begin{split}
D_{network}=\frac{1}{\underset{d=1}{\overset{M}\sum}{\left|\mathcal{S}_d\right|}}\,
\underset{d=1}{\overset{M}\sum}\,
\sum_{s\in\mathcal{S}_d}\sum_{I_{\mathcal{S}}}P\left(I_{\mathcal{S}}\right)
\underset{\textit{\textbf{Y}}_{d,\mathcal{S}_d}}\int{\left|C_{I_s}-g_{s}^1\left(\textit{\textbf{Y}}_{d,\mathcal{S}_d},I_r=f(I_\mathcal{S})\right)\right|}^2
{P\left(\textit{\textbf{Y}}_{d,\mathcal{S}_d}|I_{\mathcal{S}_d}\right)
\mathrm{d}\textit{\textbf{Y}}_{d,\mathcal{S}_d}}\textrm{,}
\end{split}
\end{equation}
where $g_{s}^1{(\textit{\textbf{Y}}_{d,\mathcal{S}_d},I_r)}$,
$s\in\mathcal{S}$ is given in Corollary \ref{C1}.
\end{corollary}

\begin{corollary}\label{C4}
In an OTN-BR with ``noiseless relay channels and very noisy
source-destination channels'', $D_{network}$ is given by:
\begin{equation}\label{E11}
\begin{split}
D_{network}&=\frac{1}{\underset{d=1}{\overset{M}\sum}{\left|\mathcal{S}_d\right|}}\,
\underset{d=1}{\overset{M}\sum}\,
\sum_{s\in\mathcal{S}_d}\sum_{I_{\mathcal{S}}}P\left(I_{\mathcal{S}}\right)
{\left|C_{I_s}-g_{s}^2\left(I_r=f(I_\mathcal{S}\right))\right|}^2\textrm{,}
\end{split}
\end{equation}
where $g_{s}^2{(I_r)}$, $s\in\mathcal{S}$ is given in Corollary
\ref{C2}.
\end{corollary}

\section{Code Design at the Relay: Implementation and Complexity Considerations}\label{SV}
As described in Section \ref{SII-C}, the relay MAP decodes the
received signals $\textit{\textbf{Y}}_{r,s}$, $s\in\mathcal{S}$ to
the indexes $\hat{I}_{\mathcal{S}}=\{\hat{I}_s: s\in\mathcal{S}\}$
and then broadcasts the network coded index
$I_r=f(\hat{I}_{\mathcal{S}})$ to the destinations. At the relay,
the network code $f(\hat{I}_{\mathcal{S}})$ is designed such that
the average reconstructed signal distortion at the destinations is
minimized. Since only $D_{network}$ depends on the network code at
the relay, the goal is to minimize $D_{network}$ as defined in
\eqref{E9}, \eqref{E10} or \eqref{E11}. To this end, an exhaustive
search to identify the optimal code at the relay requires
$2^{R_{1}}\times\ldots\times2^{R_{N}}\times2^{R_{r}}$ tests of all
combinations. Thus in the following, targeting an efficient
solution, we first devise an approach based on simulated annealing
and then inspired by the results attempt to formulate a structured
network code.

\subsection{Network Code Design at the Relay: An Approach based on Simulated Annealing}\label{SV-A}
Simulated annealing is an iterative algorithm that belongs to the
class of randomized algorithms for solving combinatorial
optimization problems. Given the current state (here network code),
the next candidate state in SA is created with certain level of
randomness, based on a so-called perturbation scheme. To avoid
sticking in local minima, a candidate state with higher cost may
also be probabilistically selected as the new state \cite{R23}. The
SA converges in probability to a global minimum if a proper
perturbation scheme and a suitably slow cooling schedule are used
\cite{R23}. The two govern the possible improvements to the code in
each iteration and the number of iterations, respectively. In
particular, if the initial temperature $T_0$ is sufficiently large,
a cooling schedule described by $T_k={c}/{\log(k+1)}$, guarantees
such a convergence \cite{R23}, where $c$ is a positive constant and
$T_k$ is the temperature after $k$ temperature drops. The SA is
previously used for index assignment in robust source coding
\cite{R23} and also for designing source and channel codes in
\cite{R24}. An alternative binary switching scheme is used in
\cite{R26} for source-optimized channel coding in point-to-point
digital transmission.

We consider a $2^{R_{1}}\times\ldots\times2^{R_{N}}$ lookup-table
(codebook) at the relay in which each element indicates the
mapping of codewords (indexes)
$\left(\hat{I}_1,\hat{I}_2,\ldots,\hat{I}_N\right)$ received from
the sources to an index at the output of the relay. The proposed
SA-based algorithm to design this code at the relay is described
below.
\begin{enumerate}
\item
 Set an initial appropriate high temperature $T=T_0$.
\item
If this is the first iteration, initialize the lookup-table
randomly. Otherwise, generate a test code by perturbing the current
one. Perturbation is accomplished by changing the value of a certain
table element to that of one random element among its adjacent
neighbors. The element for perturbation is chosen in order, e.g.,
row-wise, in subsequent iterations.
\item
Calculate $D_{network}$ using Proposition \ref{P2} (resp.
Corollary \ref{C3} or \ref{C4}). Compute the change in $D_{network}$
in comparison to that in previous iteration ($\bigtriangleup{D}$).
\item
If $\bigtriangleup{D}<0$, then the perturbed code is adopted.
Otherwise, it is only chosen with probability
$\exp(\frac{-\bigtriangleup{D}}{T})$.
\item
Iterate by going to step 2, until the code is updated for a
sufficient number of times or a maximum number of iterations is
reached.
\item
Lower the temperature. If the temperature is below a specified
value or the relative change in $D_{network}$ is insignificant,
stop; otherwise go to step 2. The cooling schedule adopted here is
$T_k=\alpha{T_{k-1}}$, $0<\alpha<1$, as in \cite{R23}\cite{R24},
which allows for a faster design process.
\end{enumerate}
When the algorithm terminates, the code design process is completed
and the resulting code may be used for operations. Thus, the main
complexity remains at the design procedure for the network code at
the relay.

\newtheorem{remark}{Remark}

\begin{remark}
To design nonlinear network codes as described, particular
assumptions on channel conditions may be made to simplify the design
procedure. This is incorporated in step 3 of the algorithm, where
Corollary \ref{C3} or \ref{C4} may be utilized instead of
Proposition \ref{P2}. This provides performance versus design
complexity trade-offs, as assessed in various channel conditions in
Section \ref{SVII}.
\end{remark}

\begin{remark}
Although, the proposed network code design algorithm is devised for a general OTN-BR,
where soft outputs may also be available at the destinations, to speed up
the design procedure, one could consider equivalent hard-decided
received signals. Naturally, in this case the integrals in
\eqref{E9}, \eqref{E10} or \eqref{E11} are replaced with summations.
\end{remark}

\newtheorem{example}{Example}

\begin{example}
As an example of the network code design, we here consider an OTN-BR with
two sources, when the source signals are Gaussian and
$R_{1}=R_{2}=R$. For $R_r=R=3$, the network code is an $8\times8$
lookup-table, where one of 8 codewords (indexed 0 to 7) is assigned
to each of its elements. One such code obtained via the proposed
SA-based scheme is presented in Fig. \ref{F2}(a). This code is
clearly nonlinear as for example, for $[I_1,I_2]=[1,1]$ and
$[\acute{I}_1,\acute{I}_2]=[2,4]$, $f(I_1,I_2)=6$ and
$f(\acute{I}_1,\acute{I}_2)=0$, while
$f(I_1\oplus\acute{I_1},I_2\oplus\acute{I_2})=0$.
Of course, the network code $f$ is linear if and only if $\forall
I_j,\acute{I}_j\in\mathcal{I}_j$, $0\leq j\leq N-1$,
$f(I_1,\ldots,I_N)\oplus
f(\acute{I_1},\ldots,\acute{I_N})=f(I_1\oplus\acute{I_1},\ldots,I_N\oplus\acute{I_N})$,
in which $\oplus$ denotes summation in $GF(2^{R_r})$. 
\end{example}

\subsection{Structured Network Code} \label{SV-B}
As stated, our experiments demonstrate that the network codes
obtained for the relay based on the proposed approach in Section
\ref{SV-A} are nonlinear. Motivated by these results, here we
attempt to formulate a structured nonlinear network code. This
provides a closed form expression of the network code at relay, and
as we shall demonstrate substantially reduces the design complexity.
It is noteworthy that there is only very limited reports in the
literature on the theory of nonlinear (channel) codes, e.g.,
\cite{R25}.

Consider the proposed network code in Fig. 2(a) designed for an
OTN-BR with two sources and $R_r=R=3$, $R_{NC}=\frac{1}{2}$. As evident, there are regions
or clusters with the same codeword which partition the codebook. Our
research indicate that the performance of the code is primarily
affected by the partitioning and not by the exact codeword (index)
assigned to each partition. Based on these observations, in Fig.
\ref{F2}(b) for the two source OTN-BR, a structured network code is
proposed based on a partitioning of the network codebook
(lookup-table at the relay). As a possible extension to the case of
OTN-BR with $N$ sources with the same network coding rate, we consider the setting where the mapping
of the first dimension and any other dimension follows the same
partitioning as that in Fig. \ref{F2}(b). The resulting family of
structured network codes $f(\hat{I}_1,\ldots,\hat{I}_N)$ may be
formulated as follows. 
\begin{equation}\label{E12}
f\left(\hat{I}_1,\ldots,\hat{I}_s,...,\hat{I}_N\right)=
\begin{cases}
a_{k_2,\ldots,k_N}&  0\leq \hat{I}_1<2 \;,\;
k_{s\neq1}=\hat{I}_s\circ
2^{R_{s}-1}:k_{s\neq1}\in\{0,1\}\\
b_{k_1,k_2,\ldots,k_N}&  2\leq \hat{I}_1<2^{R_{1}}-2,\;
k_1=\left\lfloor\frac{\hat{I}_1-2}{4}\right\rfloor,\;k_{s\neq1}=\left\lfloor
\frac{\hat{I}_s}{2^{R_{s}-2}}\right\rfloor:\\
&\;k_1\in\{0,\ldots,2^{R_{1}-2}-1\}\;,\;k_{s\neq1}\in\{0,\ldots,3\}\textrm{,}\\
e_{k_2,\ldots,k_N}&  2^{R_{1}}-2\leq \hat{I}_1<2^{R_{1}},\;
k_{s\neq1}=\hat{I}_s\circ2^{R_{s}-1}:k_{s\neq1}\in\{0,1\}
\end{cases}
\end{equation}
where, the operation $\circ$ is defined as
$\hat{I}_s\circ2^{R_{s}-1}=
\begin{cases}
0 &\text{if}\quad 0\leq \hat{I}_s<2^{R_{s}-1} \\
1&\text{otherwise}\textrm{.}
\end{cases}$. In \eqref{E12}, $a_{k_2,\ldots,k_N}$, $b_{k_1,k_2,\ldots,k_N}$,
$e_{k_2,\ldots,k_N}\in\mathcal{I}_r$ are $2^{R_r}$ distinct indices.
As the partitioning is now fixed according to \eqref{E12}, the
design is now simplified to assigning the values of these indices,
which may be handled by a SA-based algorithm with substantially
smaller complexity. Our experiments (reported in part in Section
\ref{SVII}) reveal the efficiency of the proposed code structure and
only a small performance gap with the (unstructured) code produced
by the presented algorithm in Section \ref{SV-A}.
Noting the definition of network coding rate $R_{NC}$ and utilizing
\eqref{E12} for $R_s=R$ $\forall s\in\mathcal{S}$, we have
$R_{NC}={\left\lceil\log_2\left(4\times
2^{N-2}+\left(2^{R}-4\right)\times
4^{N-2}\right)\right\rceil}/{NR}$, which as desired is nearly
$\frac{1}{2}$ for $R\in\{2,3,4\}$ and $N<5$.

\section{Separate Source Channel-Network Coding Bound}\label{SVI}
In this section, we present performance bounds by combining the
rate-distortion function due to source coding with the capacity
upper-bound, due to the cut-set bound for wireless networks under
consideration. The latter corresponds to joint channel-network
coding. Therefore, we refer to the bounds thus derived as the
separate source and channel-network coding bound, which is naturally
obtained assuming large block lengths.

In the following, we first present a cut-set upper bound on the
achievable rates $R_{c,s}$ (bits per channel use - due to joint
channel-network coding) in a Gaussian OTN-BR for communications
between the source node $s$ and the corresponding destination nodes,
$\mathcal{D}_s$. In a Gaussian OTN-BR, each link is modeled as an independent additive white
Gaussian noise channel, where the noise is zero mean with variance
$\sigma^2_{n,m},\forall(n,m)\in\mathcal{G}$ (refer to \eqref{E1} for
the definition of $\mathcal{G}$). Considering orthogonality of channels based on time
multiplexing, the sources and the relay each transmit in a fraction
of time $T$, i.e., $\lambda_1T,\ldots,\lambda_NT$ and
$\bigl(1-\underset{s\in\mathcal{S}}\sum{\lambda_s}\bigr)T$,
respectively. Suppose that the sources and the relay are subject to the
average power constraints $P_1,\ldots, P_N$ and $P_r$. Thus, these
nodes can transmit respectively with the average powers
$\acute{P_1}={P_1}/{\lambda_1},\ldots,
\acute{P_N}={P_N}/{\lambda_N}$ and
$\acute{P_r}={P_r}/{\bigl(1-\underset{s\in\mathcal{S}}\sum{\lambda_s}\bigr)}$
during their corresponding transmission periods.

\begin{proposition}\label{P3}
For a Gaussian OTN-BR with $N$ source and $M$ destination nodes,
the cut-set upper bound for the rates
$(R_{c,1},\ldots,R_{c,s},\ldots,R_{c,N})$ is given by
\begin{equation}\label{E13}
\begin{split}
&\underset{\substack{s\in\mathcal{F}:\\\mathcal{F}\subset\mathcal{S}}}\sum{R_{c,s}}\leq
\underset{\substack{0\leq\lambda_{\acute{s}}\leq1\\{\acute{s}\in\mathcal{S}}
}}{\operatorname{max}}\;\operatorname{min}
\left\{\underset{\acute{s}\in\mathcal{F}}\sum{\lambda_{\acute{s}}}
\frac{1}{2}\log_2\left(1+\hspace{-2pt}\hspace{-2pt}\underset{d\in\{\mathcal{D}_{\acute{s}}\}}\sum{\hspace{-2pt}\frac{\acute{P}_{\acute{s}}}{\sigma^2_{\acute{s},d}}}+\frac{\acute{P}_{\acute{s}}}{\sigma^2_{\acute{s},r}}\right),\underset{\acute{s}\in\mathcal{F}}\sum{\lambda_{\acute{s}}}\frac{1}{2}\log_2\left(1+\hspace{-2pt}\underset{d\in\{\mathcal{A}\cap\mathcal{D}_{\acute{s}}\}}
\sum{\frac{{\acute{P}_{\acute{s}}}}{\sigma^2_{{\acute{s}},d}}}\right)\right.\\
&\left.+\left(1-\underset{{\acute{s}}\in\mathcal{S}}\sum{\lambda_{\acute{s}}}\right)
\frac{1}{2}\log_2\left(1+\underset{d\in\mathcal{A}}\sum{\frac{\acute{P_r}}{\sigma^2_{r,d}}}\right):
\forall\mathcal{A}\subset\mathcal{D_\mathcal{F}},\;\mathcal{S}_{\mathcal{A}}\cap\mathcal{F}=\mathcal{F}\right\}
\end{split}
\end{equation}
\end{proposition}\vspace{-2pt}
\begin{proof}
We consider $\mathcal{F}$ as a subset of sources, i.e.,
$\mathcal{F}\subset \mathcal{S}$. A set of destinations, whose every member intends to decode all the sources in $\mathcal{F}$ is
denoted by $\mathcal{A}$, $\mathcal{A}\subset \mathcal{D}_\mathcal{F}$. The cut $C_1$, is considered as that crossing only the outgoing channels from $\mathcal{F}$. On the other hand, the cut crossing the incoming channels to $\mathcal{A}$ and also including only the outgoing channels from the relay and $\mathcal{F}$ is denoted by $C_2$. Using max-flow min-cut theorem for the OTN-BR, we see that the maximum transmission sum-rate by the sources in $\mathcal{F}$ is equal to the minimum of information flow across the cut $C_1$ or the cuts $C_2$ corresponding to all possible subsets $\mathcal{A}$. As evident in RHS of \eqref{E13},
the first term corresponds to the cut $C_1$ and the next term is
related to the cuts $C_2$. The complete proof of Propositions \ref{P3}
is provided in Appendix \ref{A2}.
\end{proof}

For a Gaussian source with variance $\sigma_s^2$ and a source coding
rate of $R_s$ bits per source sample, the distortion-rate function
$D(R_s)$ is equal to $\sigma^2_s 2^{-2R_s}$ \cite{R6}. Considering
separate source and channel-network coding, we have
\begin{equation}\label{E25}
D(bR_{c,s})=\sigma^2_s 2^{-2bR_{c,s}},
\end{equation}
where $b$ denotes the number of channel uses per source sample.
Thus, equation \eqref{E25} in conjunction with Proposition 3,
presents a distortion-power function that serves as a performance
bound in the sequel. \vspace{-2pt}

We now consider Proposition \ref{P3} in a special case for a Gaussian OTN-BR-(2,3) in a
symmetric network setting. A Gaussian OTN-BR is symmetric
when $\sigma^2_{s,d}=\sigma^2_{\mathscr{s}\mathscr{d}}$,
$\sigma^2_{s,r}=\sigma^2_{\mathscr{s}\mathscr{r}}$ ,
$\sigma^2_{r,d}=\sigma^2_{\mathscr{r}{\mathscr{d}}}$ and $P_s=P$,
$s\in\mathcal{S},\; d\in\mathcal{D}$. We thus have $R_{c,s}=R_c$ $\forall
s\in\mathcal{S}$.
\begin{corollary}\label{C5}
For a symmetric Gaussian OTN-BR-(2,3), the cut-set upper bound on
the rates $R_{c,s}=R_c,\;s\in\mathcal{S}$ is given by \vspace{-2pt}
\begin{equation*}\label{E14}
\begin{split}
R_c\leq\underset{0\leq\lambda\leq\frac{1}{2}}{\operatorname{max}}\;\operatorname{min}&\left\{\lambda\frac{1}{2}\log_2\left(1+\frac{2P}{\lambda\sigma^2_{\mathscr{s}\mathscr{d}}}+\frac{P}{\lambda\sigma^2_{\mathscr{s}\mathscr{r}}}\right),\right.
\end{split}
\end{equation*}
\begin{equation}\label{E14}
\begin{split}
\left.\lambda\frac{1}{2}\log_2\left(1+\frac{P}{\lambda\sigma^2_{\mathscr{s}\mathscr{d}}}\right)+(1-2\lambda)\frac{1}{4}\log_2\left(1+\frac{P_r}{(1-2\lambda)\sigma^2_{\mathscr{r}\mathscr{d}}}\right)\right\}\textrm{.}
\end{split}
\end{equation}
\end{corollary}
Following the approach described above, we can use \eqref{E25} and
obtain a distortion-power function for a symmetric Gaussian
OTN-BR(2,3). This provides a performance bound for comparisons in
Section \ref{SVII}. The bounds for OMARC and OMC-BR can be obtained
in a similar way.

\section{Performance Evaluation }\label{SVII}
Consider a symmetric Gaussian OTN-BR with $N$ source nodes producing
independent Gaussian distributed signals with zero mean and unit
variance, that are to be transmitted to the corresponding
destinations. Each of these continuous signals is quantized using a
$2^{R}$ level Lloyd-Max quantizer. The resulting quantization index
is represented by a binary codeword. Each binary codeword is BPSK
modulated and transmitted through a Gaussian channel. The network
code at the relay is described by a $2^{R}\times2^{R}$ and
$2^{R}\times2^{R}\times2^{R}$ lookup-table for $N=2$ and $N=3$,
respectively. To obtain network codes of rate (approximately)
$\frac{1}{2}$, the relay transmission rate is selected as $R_r=R$
for $N=2$, and $R_r=5$ bits for $N=3$, $R=3$.

At the relay, the received codewords from the sources are (symbol)
MAP decoded and are combined using the proposed nonlinear network
coding. Hence, we refer to them as decode and nonlinear network
coding (DNNC) schemes and denote them by DNNC-Structured, DNNC-C3
and DNNC-C4 indicating, as described in Section \ref{SV}, the design
based on the structured network code, or Corollaries \ref{C3} and
\ref{C4}, respectively. For comparison, we consider MAP decoding
followed by classic linear network coding at the relay and refer to
it as decode and linear network coding (DLNC). In this case, the
binary codewords are represented and linearly combined in
$GF(2^{R_r})$ with coefficients searched and selected to minimize
the average distortion at the destinations. The code thus obtained
is referred to as the best performing linear network code.

For the performance evaluations of symmetric OTN-BRs in the
following, the signal to noise ratio of the channel from any source
to any of its corresponding destination nodes is denoted by $
S\!N\!R^{design}_{\mathscr{s}\mathscr{d}}$ and
$S\!N\!R_{\mathscr{s}\mathscr{d}}$, during the design and
operations, respectively. In the same direction,
$S\!N\!R_{\mathscr{r}\mathscr{d}}$ and
$S\!N\!R_{\mathscr{s}\mathscr{r}}$, respectively demonstrate the
SNRs of the relay-destination channels and the source-relay channels
during the operations. The average reconstruction signal SNR (RSNR),
i.e., $\dfrac{1}{N}\underset{i=1}{\overset{N}\sum}
\dfrac{E\left[X_i^2\right]}{D}$, where $D$ is obtained from
\eqref{E7}, is used as the performance criterion.

\subsection{Basic Comparisons}
 Figs. \ref{F3} and \ref{F4} respectively, present
the performance of an OMC-BR and an OMARC, with two sources, $R=3$,
$R_{NC}=\frac{1}{2}$ and $S\!N\!R_{\mathscr{s}\mathscr{d}}=\rm
-3dB$. These figures and the results for the OTN-BR-(2,3) (not
reported here) demonstrate that the proposed DNNC schemes
substantially outperform the DLNC scheme. Specifically, a RSNR gain of
about 4dB is achieved for all considered networks, when
the source-relay and the relay-destination channels are of good
quality.

The DNNC-C4 (with the code in Fig. \ref{F2}(a)) is designed assuming very noisy source-destination channels and is simpler to design in comparison
to DNNC-C3. Both DNNC-C3 and DNNC-C4, are designed assuming noiseless relay-destination channels. In such settings, however, as evident in Figs. \ref{F3} to \ref{F5}, DNNC-C3 outperforms DNNC-C4.

Examining the performance of the proposed DNNC-C3 scheme designed
for $S\!N\!R_{\mathscr{s}\mathscr{d}}^{design}$ of -3dB and 1dB in
Figs. \ref{F3} to \ref{F5} demonstrates that, as expected, when the
operating channel conditions are closer to the design setting, a
higher performance is obtained. It is interesting to note, however,
that the sensitivity of the performance to a mismatch of design and
operations channel SNRs is insignificant. Our experiments over
severely asymmetric networks indicate that the proposed network code
in DNNC-C3 is designed to further assist the communication of
source(s) with lower source-destination channel SNR.

Also, comparing Figs. 4 and 5, it is evident that when the quality
of the source-destination channels improves, as expected, the RSNR
is enhanced and ultimately reaches that of a 3-bit Lloyd Max
quantizer ($\sim15$dB). For noisy source-destination channels, the
RSNR does not reach this limit even if the source-relay and the
relay-destination channels are noiseless. This is due to the fact
that the network coding at the relay is not a one-to-one mapping,
and as expected, the corresponding relay signal is only meant to
assist the source-destination transmission. In the current example,
the relay receives $R_1=3$ and $R_2=3$ bit codewords from sources 1
and 2, respectively, while produces only one $R_r=3$ bit codeword.

In our studies, we have also examined another scheme (not reported
here), which is referred to as estimate and forward nonlinear
network coding (ENNC). In ENNC, the source signal transmitted by
each source is first estimated at the relay using an MMSE decoder.
Subsequently, the estimated source signals at the relay are mapped
to an output codeword using an optimized vector quantization scheme.
According to our simulations, the ENNC provides only a limited gain
over the DNNC at the cost of increased complexity.

\subsection{Effects of Rates of Network Coding and Quantization}
Fig. \ref{F6} presents the performance of OMC-BR network with $N=3$,
for different network coding schemes and rates. The operation
source-destination channel SNRs are set to
$S\!N\!R_{\mathscr{s}\mathscr{d}}=\rm-3dB$. It is observed that,
with $R=3$, $R_r=5$ and $R_{NC}=\frac{1}{2}$, the proposed DNNC
schemes achieve a gain of about 4dB in the RSNR compared to DLNC.
Therefore, noting Figs. \ref{F3}, \ref{F4} and \ref{F6}, it is
evident that similar gains are achieved in the considered networks
with $N=2$ and $N=3$ for a given source-destination channel SNR, and
rates of network coding and quantization. Fig. \ref{F6} also depicts
the performance of the DLNC and the proposed DNNC schemes for
$R_r=R=3$, $R_{NC}=\frac{1}{3}$. As expected, for a given
quantization rate, the performance gain provided by the DNNC schemes
is greater, when the network code rate $R_{NC}$ is larger.

Figs. \ref{F7} and \ref{F8}, respectively depict the resulting
average distortion in an OMARC and an OTN-BR-(2,3) as a function of
rate for $R\in\{2,3,4\}$. It is observed that the proposed DNNC
scheme provides a larger performance gain with respect to DLNC, when
the quantization rate is greater.

\subsection{Performance of the Proposed Structured Nonlinear Network Code}
The performance of the proposed structured nonlinear network code is
depicted in Figs. \ref{F3} to \ref{F6}. This code is presented in
Fig. \ref{F2}(b) for $N=2$ and in \eqref{E12} for arbitrary $N$.
As evident, the structured network code performs closely similar to
DNNC-C3 especially for high quality relay-destination channels.
However, this performance is obtained with much smaller design
complexity. Our experiments indicate that the nearly symmetric
structure of the proposed structured network code leads to almost
identical performance at different receivers over a symmetric
network. For example, in a symmetric OMC-BR with $N=3$,
$S\!N\!R_{\mathscr{s}\mathscr{d}}=\rm-3dB$ and
$S\!N\!R_{\mathscr{r}\mathscr{d}}=\rm7dB$, the average RSNR at the
destinations 1, 2 and 3 is 8.42dB, 8.74dB and 8.84dB, respectively.

Considering symmetric network settings and the definition of the
average distortion in \eqref{E7}, each of the proposed DNNC schemes
in the three OTN-BR instances of interest provides the same
performance for given $S\!N\!R_{\mathscr{s}\mathscr{d}}$ and
$S\!N\!R_{\mathscr{r}\mathscr{d}}$. This issue is evident in the
simulation results depicted in Figs. \ref{F3}-\ref{F5}, \ref{F7} and
\ref{F8}. However, DLNC provides a better performance in OMARC in
comparison to other considered networks as it is seen in Figs.
\ref{F7} and \ref{F8}.

\subsection{Distortion Power Trade-off}
Fig. \ref{F9} demonstrates the trade-off of average distortion and
power of the sources ($S\!N\!R_{\mathscr{s}\mathscr{d}}$) for an
OTN-BR-(2,3). In this figure, the performance of the proposed DNNC
with structured nonlinear network coding and DLNC are depicted
together with the separate source and channel-network coding bound
obtained from Corollary \ref{C5} for comparison. In Fig. \ref{F9},
it is seen that the proposed DNNC with structured network code
compared to DLNC reduces the gap to the bound by more than $50\%$.
It is also evident in this figure that, (i) for high
$S\!N\!R_{\mathscr{s}\mathscr{d}}$ both DNNC and DLNC result in
identical residual distortion equal to that of a 3-bit Lloyd-Max
quantizer, that is due to source coding; and (ii) an improved
$S\!N\!R_{\mathscr{r}\mathscr{d}}$ enhances the performance gain
provided by the proposed DNNC scheme. Our simulation results (not
reported here) confirm similar observations for the OMARC and
OMC-BR.

\section{Conclusions}\label{SVIII}
In this paper, network code design for the orthogonal two-hop
network with broadcasting relay and constrained complexity was
investigated. Taking a joint source-channel-network coding approach,
the network code at the relay was designed to minimize the
average distortion at the destinations. Decomposing the distortion into its components enabled the development of an
effective network code design algorithm based on simulated annealing. The
resulting network code is nonlinear and outperforms the best
performing linear network codes. This indicates the insufficiency of
linear network coding in complexity constrained wireless networks
with a MMSE design criteria. The results also show that the
sensitivity of the proposed nonlinear network code to a mismatch of
design and operations channel SNRs is insignificant and the
performance gain provided by the nonlinear code compared to the linear
code is greater, when the network code rate is larger. In comparison
to the separate source and channel-network coding bound, the
proposed nonlinear network coding at the relay in contrast to the linear code
reduces the gap to the bound by more than $50\%$. This fact
indicates the effectiveness of the proposed decode and nonlinear network coding schemes for the OTN-BR and a
promising research direction.


%

\appendices
\section{Proof of Proposition \ref{P1}}\label{A1}
At the destination $d\in\mathcal{D}$, based on the received signals
$\textit{\textbf{Y}}_{d,\mathcal{S}_d}$ and
$\textit{\textbf{Y}}_{d,r}$, the MMSE estimation of the transmitted
signal $X_s,\;s\in\mathcal{S}_d$ is given by:
\begin{equation}\label{E15}
\hat{x}_{s,d}=g_s\left(\textit{\textbf{Y}}_{d,\mathcal{S}_d},\textit{\textbf{Y}}_{d,r}\right)=E\left[X_s|\textit{\textbf{Y}}_{d,\mathcal{S}_d},\textit{\textbf{Y}}_{d,r}\right]=\underset{I_{\mathcal{S}_d}}{\sum}E\left[X_s|I_{\mathcal{S}_d}\right]P\left(I_{\mathcal{S}_d}|\textit{\textbf{Y}}_{d,\mathcal{S}_d},\textit{\textbf{Y}}_{d,r}\right).
\end{equation}
Since the sources are independent, we have
$E\left[X_s|I_{\mathcal{S}_d}\right]=E\left[X_s|I_s\right]$ and from
\eqref{E15} we obtain
\begin{equation}\label{E16}
g_s\left(\textit{\textbf{Y}}_{d,\mathcal{S}_d},\textit{\textbf{Y}}_{d,r}\right)=\underset{I_{\mathcal{S}_d}}{\sum}E\left[X_s|I_s\right]P\left(I_{\mathcal{S}_d}|\textit{\textbf{Y}}_{d,\mathcal{S}_d},\textit{\textbf{Y}}_{d,r}\right){.}
\end{equation}
Using Bayes' theorem, we have
\begin{equation}\label{E17}
P\left(I_{\mathcal{S}_d}|\textit{\textbf{Y}}_{d,\mathcal{S}_d},\textit{\textbf{Y}}_{d,r}\right)=\frac{P\left(\textit{\textbf{Y}}_{d,\mathcal{S}_d},\textit{\textbf{Y}}_{d,r}|I_{\mathcal{S}_d}\right)P(I_{\mathcal{S}_d})}{P\left(\textit{\textbf{Y}}_{d,\mathcal{S}_d},\textit{\textbf{Y}}_{d,r}\right)},
\end{equation}
\vspace{-2pt}where\vspace{-2pt}
\begin{equation}\label{E18}
P\left(\textit{\textbf{Y}}_{d,\mathcal{S}_d},\textit{\textbf{Y}}_{d,r}|I_{\mathcal{S}_d}\right)=P\left(\textit{\textbf{Y}}_{d,\mathcal{S}_d}|I_{\mathcal{S}_d}\right)P\left(\textit{\textbf{Y}}_{d,r}|I_{\mathcal{S}_d}\right).
\end{equation}
Thus, using \eqref{E17} and \eqref{E18}, equation \eqref{E16} is
rewritten as
\begin{equation}\label{E19}
g_s\left(\textit{\textbf{Y}}_{d,\mathcal{S}_d},\textit{\textbf{Y}}_{d,r}\right)=\frac{\underset{I_{\mathcal{S}_d}}{\sum}E\left[X_s|I_s\right]P\left(\textit{\textbf{Y}}_{d,\mathcal{S}_d}|I_{\mathcal{S}_d}\right)P\left(\textit{\textbf{Y}}_{d,r}|I_{\mathcal{S}_d}\right)P(I_{\mathcal{S}_d})}{P\left(\textit{\textbf{Y}}_{d,\mathcal{S}_d},\textit{\textbf{Y}}_{d,r}\right)},
\end{equation}
\vspace{-4pt}where\vspace{-4pt}
\begin{equation}\label{E20}
P\left(\textit{\textbf{Y}}_{d,r}|I_{\mathcal{S}_d}\right)=\underset{I_r}{\sum}P\left(\textit{\textbf{Y}}_{d,r}|I_r\right)P\left(I_r|I_{\mathcal{S}_d}\right){.}
\end{equation}
Finally, using \eqref{E19} and \eqref{E20}, we can obtain
\eqref{E2}. Based on \eqref{E18}, the term
$K=P\left(\textit{\textbf{Y}}_{d,\mathcal{S}_d},\textit{\textbf{Y}}_{d,r}\right)$
is given by \eqref{E3} as a factor normalizing the sum of
probabilities to one.

\section{Proof of Proposition \ref{P3}}\label{A2}
Consider $\mathcal{D}_{\mathcal{F}}=\underset{s\in\mathcal{F}}\cup
\mathcal{D}_s$,
$\mathcal{S}_{\mathcal{A}}=\underset{d\in\mathcal{A}}{\cup}\mathcal{S}_d$
and
$\textit{\textbf{Y}}_{\mathcal{A},\mathcal{F}}=\left\{\textit{\textbf{Y}}_{d,s}:d\in\mathcal{A},\;s\in\mathcal{F},\;(s,d)\in\mathcal{G}\right\}$,
$\forall{\mathcal{F}}\subset{S}$ and
$\forall{\mathcal{A}}\subset{\mathcal{D}_\mathcal{F}}$, where
$\mathcal{G}$ is defined in \eqref{E1}. In view of the max-flow
min-cut theorem \cite{R6} for the OTN-BR, an outer bound for the
capacity region is the set of rate vectors $(R_{c,1},R_{c,2},\ldots
R_{c,N})$ satisfying
\begin{equation}\label{E21}
\begin{split}
&\underset{s\in\mathcal{F}}{\sum}R_{c,s}\leq\operatorname{min}\left\{I\left(X_{\mathcal{F}};\textit{\textbf{Y}}_{\mathcal{D}_\mathcal{F},\mathcal{F}},\textit{\textbf{Y}}_{r,\mathcal{F}}|X_{\mathcal{F}^c},X_r,Q\right),\right.\\
&\left.I\left(X_{\mathcal{F}},X_r;\textit{\textbf{Y}}_{\mathcal{A},r},\textit{\textbf{Y}}_{\mathcal{A},\mathcal{F}}|X_{\mathcal{F}^c},Q\right):\;\forall{\mathcal{A}}\subset{\mathcal{D}_{\mathcal{F}}},\;\mathcal{S}_{\mathcal{A}}\cap\mathcal{F}=\mathcal{F}\right\}
\end{split}
\end{equation}
over all distributions $P(Q)\left(\prod_{s=1}^N
P(X_s|Q)\right)P\left(X_r|X_{\mathcal{S}},Q\right)$ with $|Q|\leq
N+1$. Considering time-multiplexing in OTN-BR as described before,
we have
\begin{equation}\label{E22}
\begin{split}
&\underset{s\in\mathcal{F}}{\sum}R_{c,s}\leq
\underset{\substack{0\leq\lambda_{\acute{s}}\leq1\\
{\acute{s}}\in\mathcal{S}}}{\operatorname{max}}\;\operatorname{min}\left\{\underset{{\acute{s}}\in\mathcal{F}}{\sum}\lambda_{\acute{s}}
I\left(X_{\acute{s}};\textit{\textbf{Y}}_{\mathcal{D}_{\acute{s}},\{{\acute{s}}\}},\textit{\textbf{Y}}_{r,{\acute{s}}}|Q\right),\right.\\
&\left.\underset{{\acute{s}}\in\mathcal{F}}{\sum}\lambda_{\acute{s}}
I\left(X_{\acute{s}};\textit{\textbf{Y}}_{\mathcal{A},\{{\acute{s}}\}}|Q\right)+\left(1-\underset{{\acute{s}}\in\mathcal{S}}{\sum}\lambda_{\acute{s}}\right)I\left(X_r;\textit{\textbf{Y}}_{\mathcal{A},\{r\}}|Q\right):\;\forall{\mathcal{A}}\subset{\mathcal{D}_\mathcal{F}},\;\mathcal{S}_{\mathcal{A}}\cap\mathcal{F}=\mathcal{F}\right\}\textrm{.}
\end{split}
\end{equation}

For Gaussian channels, the mutual information terms in \eqref{E22}
are maximized when the distributions of
$X_k,\;k\in\left\{\mathcal{S}\cup\left\{r\right\}\right\}$ are
Gaussian. Noting that in the considered settings we have
$\underset{P(X_k)}{\operatorname{max}}I\left(X_k;\textit{\textbf{Y}}_{\mathcal{B},\{k\}}\right)=\dfrac{1}{2}\log_2\left(1+\sum\limits_{\substack{j\in\mathcal{B},\\(k,j)\in\mathcal{G}}}\dfrac{P_k}{\sigma^2_{k,j}}\right)$,
where $\mathcal{B}\subset\{\mathcal{D}\cup\{r\}\}$ and
$k\in\left\{\mathcal{S}\cup\left\{r\right\}\right\}$ \cite{R6}, the
proof of Proposition \ref{P3} is complete.

\ifCLASSOPTIONcaptionsoff
  \newpage
\fi



\bibliographystyle{IEEEtran}
\bibliography{IEEEabrv,Ref}
\begin{figure}[h!]
\centering
\includegraphics[width=4in]{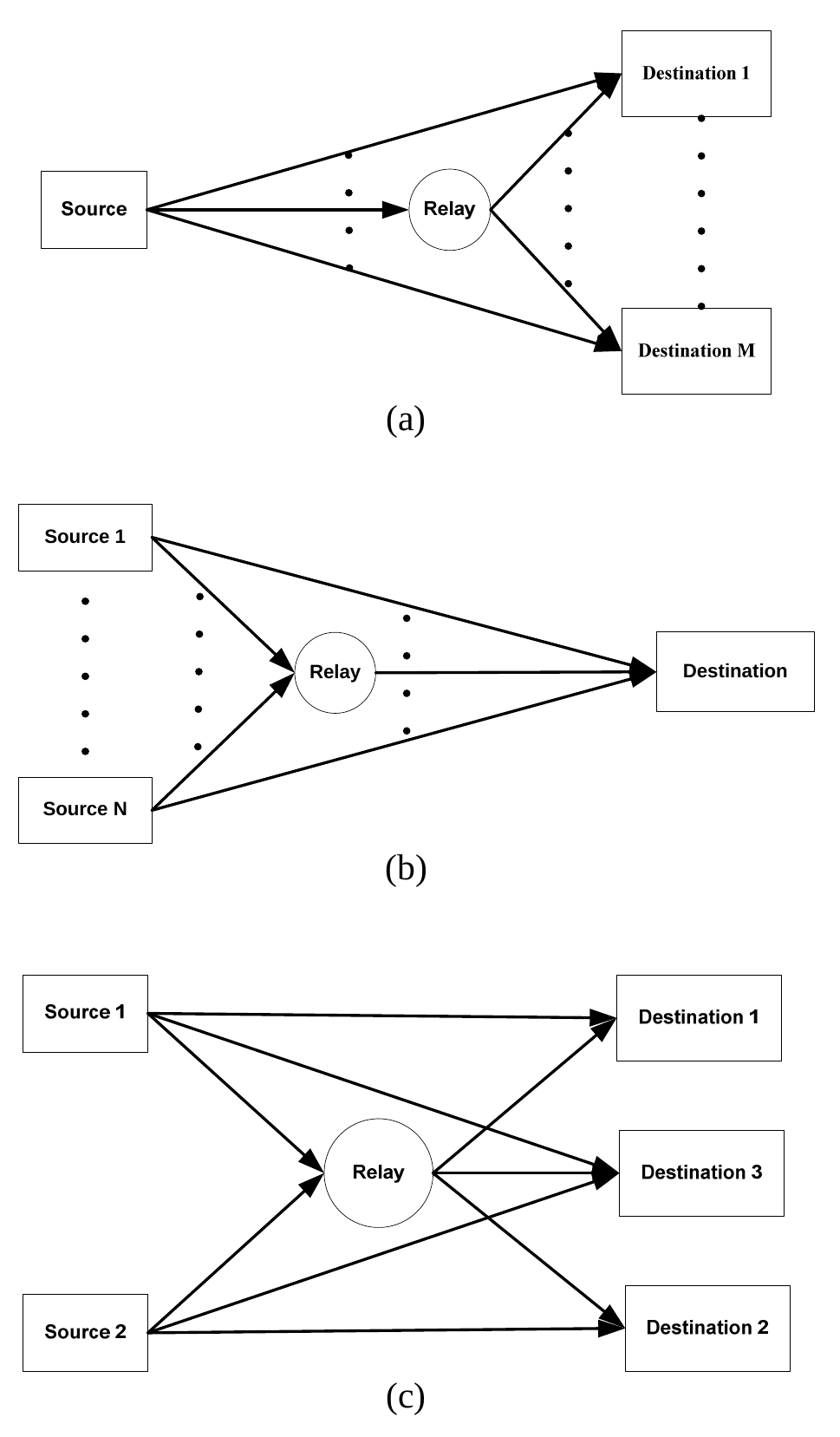}
\caption{(a) BRC (b) MARC and (c) TN-BR-(2,3).}\label{F1}
\end{figure}

\begin{figure}[h!]
\centering
\includegraphics[width=4in]{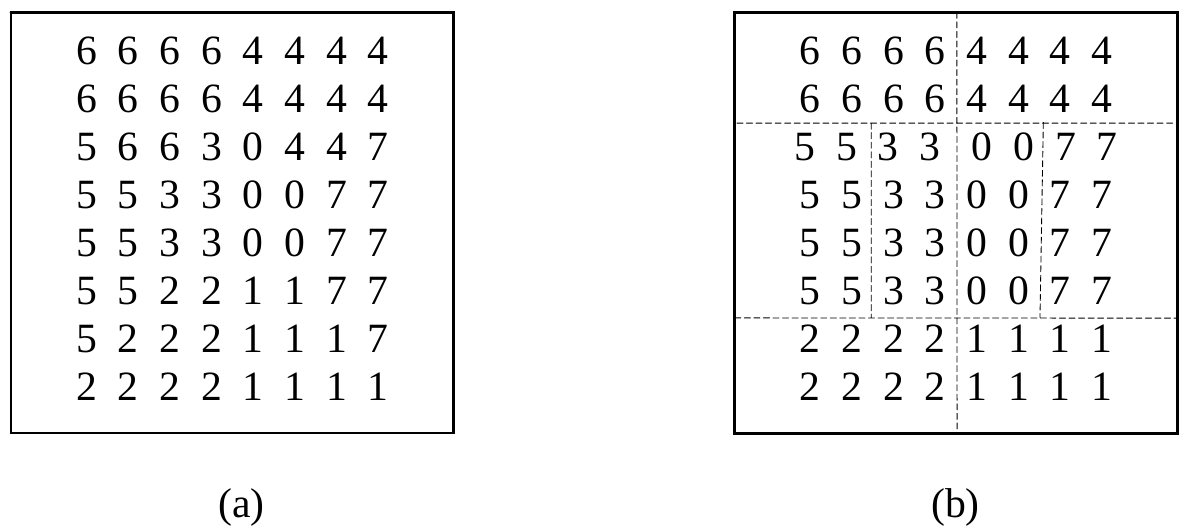}
\caption{Nonlinear network codes for an OTN-BR with $N=2$ and
$S\!N\!R_{\mathscr{s}\mathscr{r}}=\rm10dB$, (a) using DNNC-C4 and
(b) structured network coding.}\label{F2}
\end{figure}

\begin{figure}[h!]
\centering
\includegraphics[width=4in]{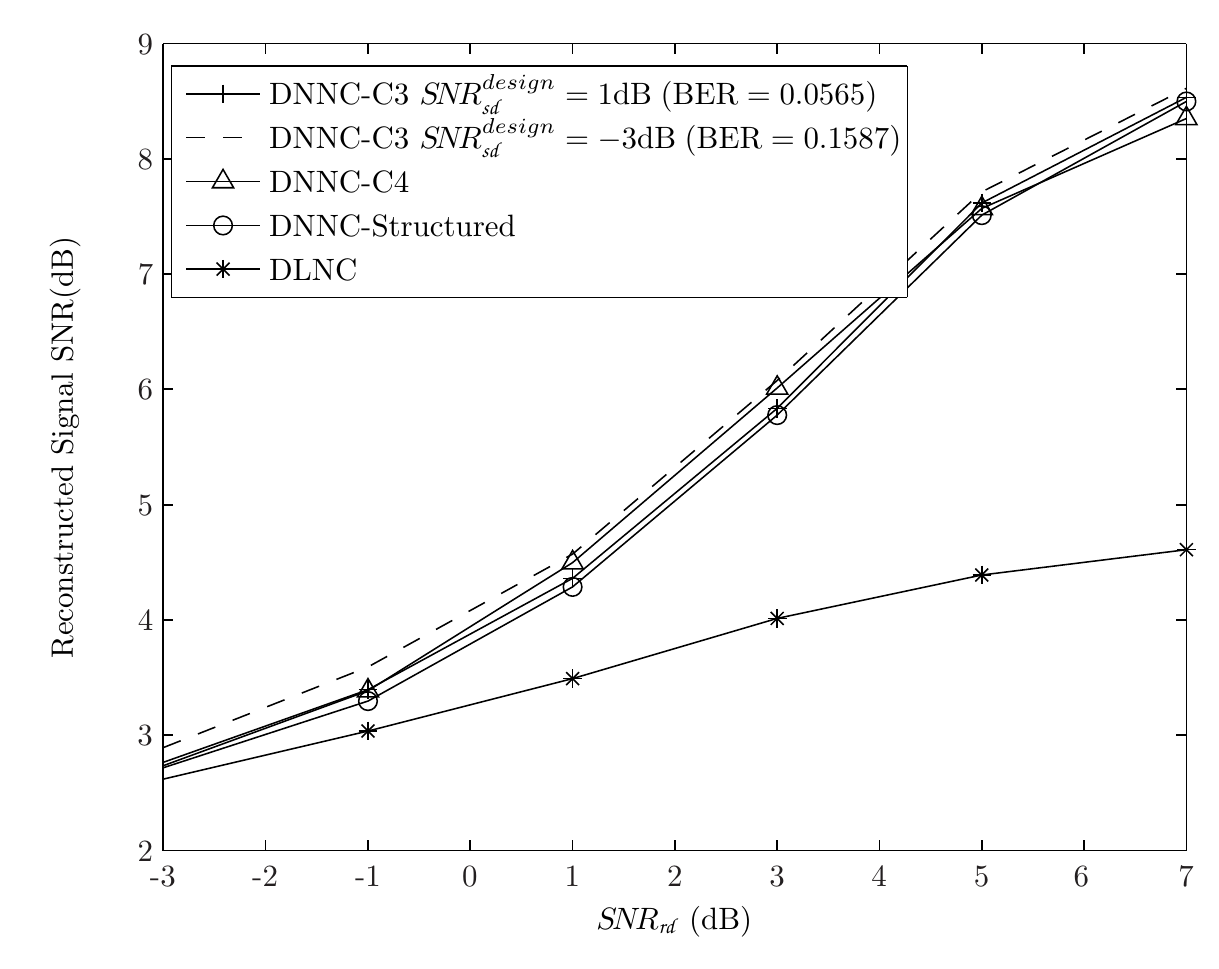}
\caption{Performance of the proposed DNNC and DLNC for an OMC-BR
with two sources, $R_r=R=3$,
$S\!N\!R_{\mathscr{s}\mathscr{r}}=\rm10dB$ and
$S\!N\!R_{\mathscr{s}\mathscr{d}}=\rm -3dB$.}\label{F3}
\end{figure}

\begin{figure}[h!]
\centering
\includegraphics[width=4in]{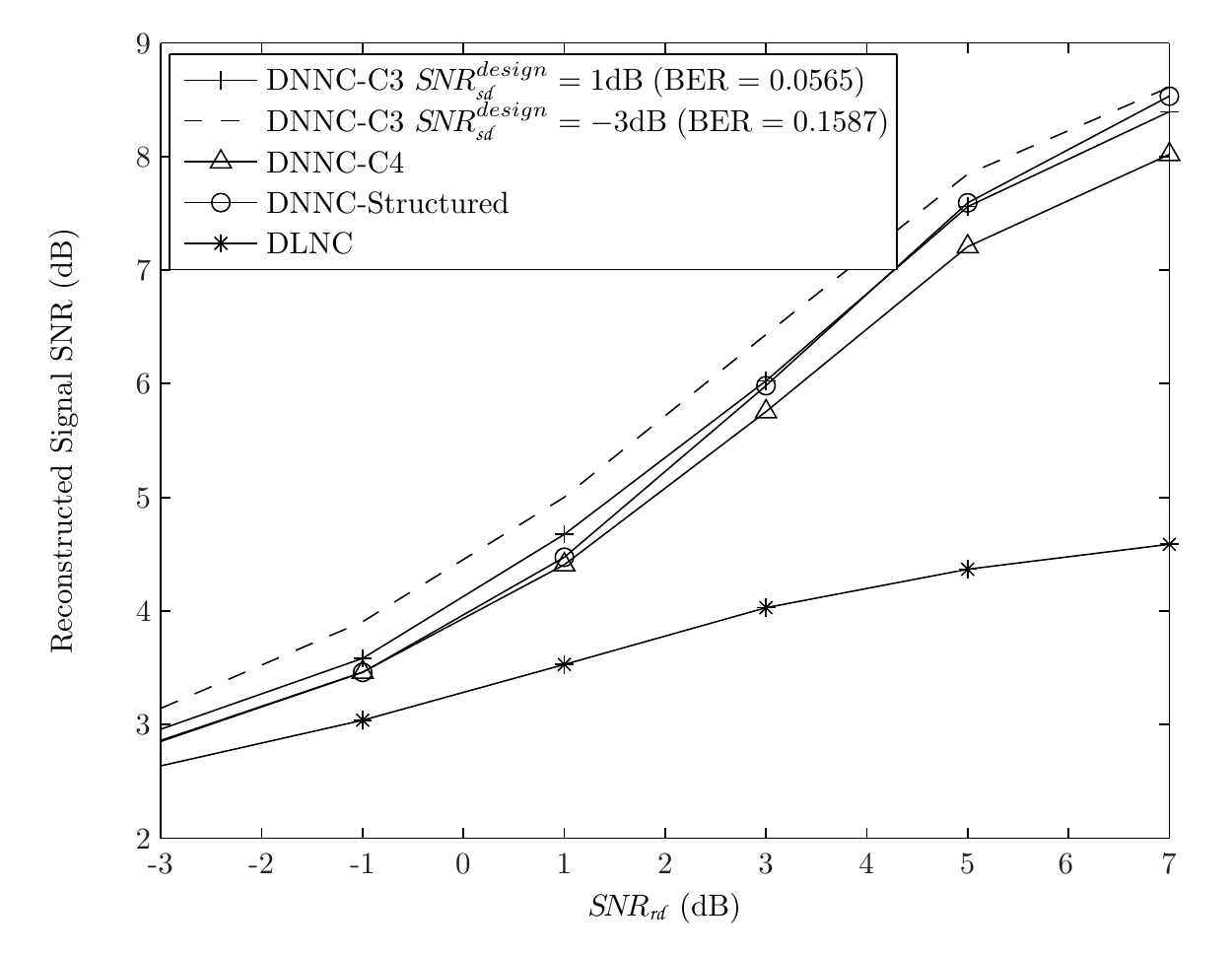}
\caption{Performance of the proposed DNNC and DLNC for an OMARC with
two sources, $R_r=R=3$, $S\!N\!R_{\mathscr{s}\mathscr{r}}=\rm 10dB$
and $S\!N\!R_{\mathscr{s}\mathscr{d}}=\rm -3dB$.}\label{F4}
\end{figure}

\begin{figure}[h!]
\centering
\includegraphics[width=4in]{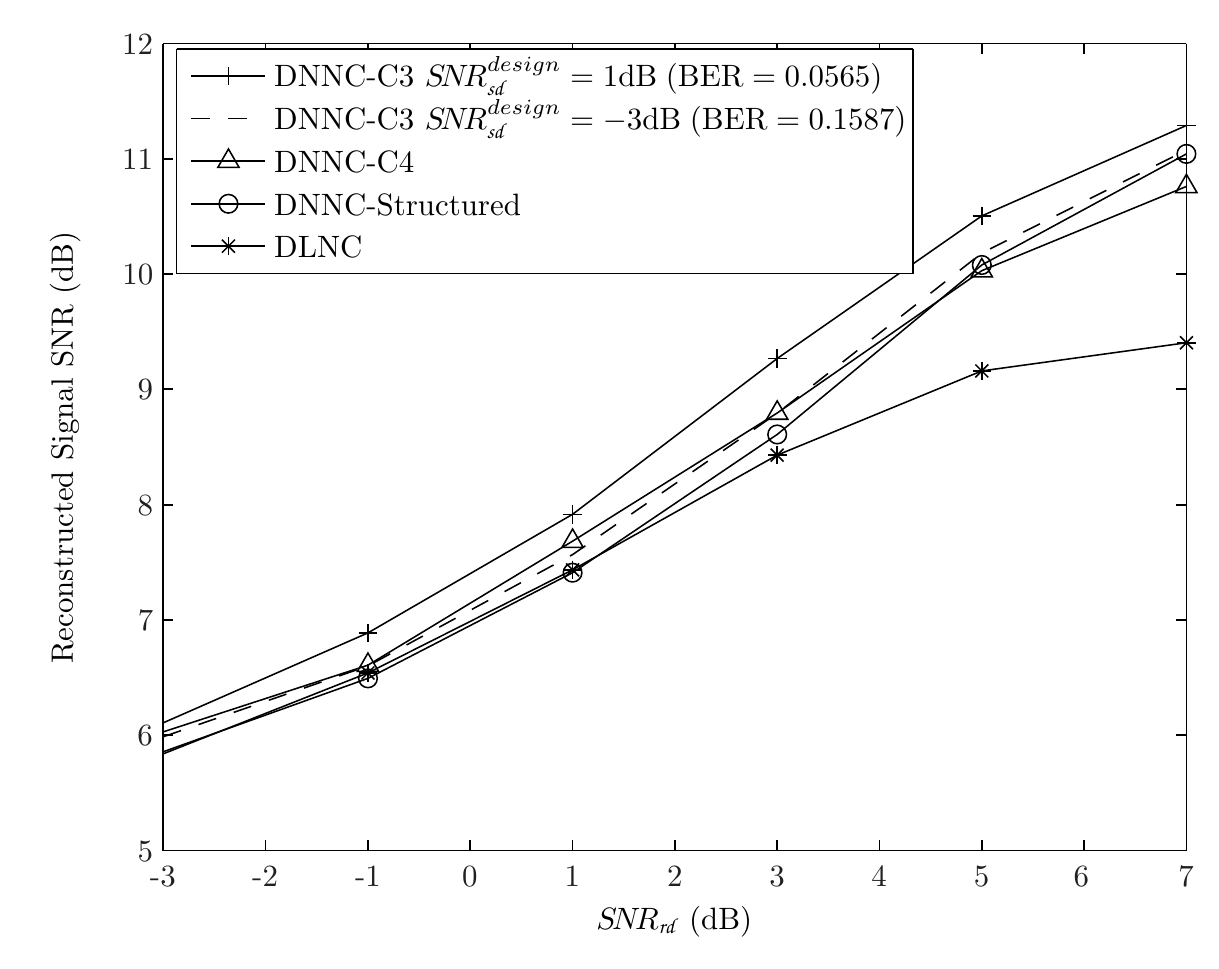}
\caption{Performance of the proposed DNNC and DLNC for an OMARC with
two sources, $R_r=R=3$, $S\!N\!R_{\mathscr{s}\mathscr{r}}=\rm 10dB$
and $S\!N\!R_{\mathscr{s}\mathscr{d}}=\rm 1dB$.}\label{F5}
\end{figure}

\begin{figure}[h!]
\centering
\includegraphics[width=4in]{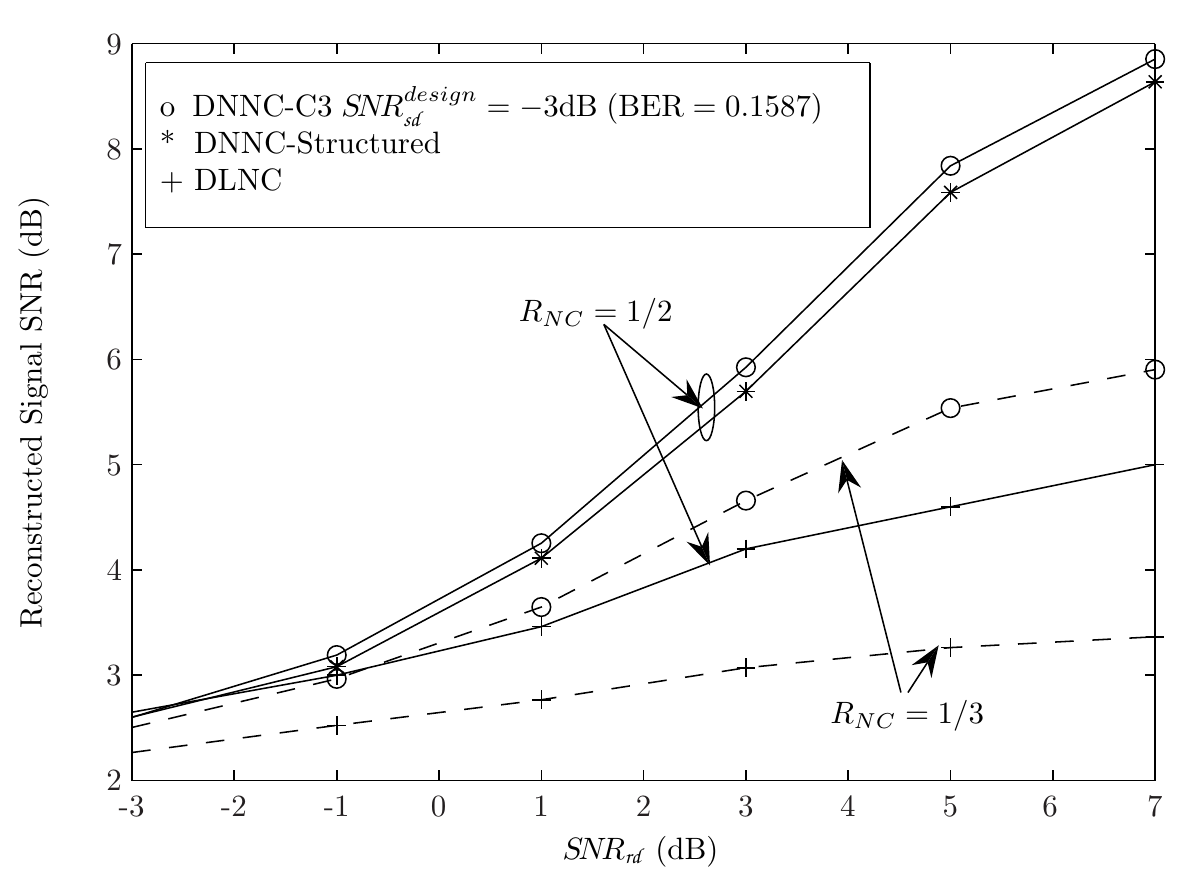}
\caption{Performance of the proposed DNNC and DLNC for an OMC-BR
with three sources, $S\!N\!R_{\mathscr{s}\mathscr{r}}=\rm 10dB$ and
$S\!N\!R_{\mathscr{s}\mathscr{d}}=\rm -3dB$.}\label{F6}
\end{figure}

\begin{figure}[h!]
\centering
\includegraphics[width=4in]{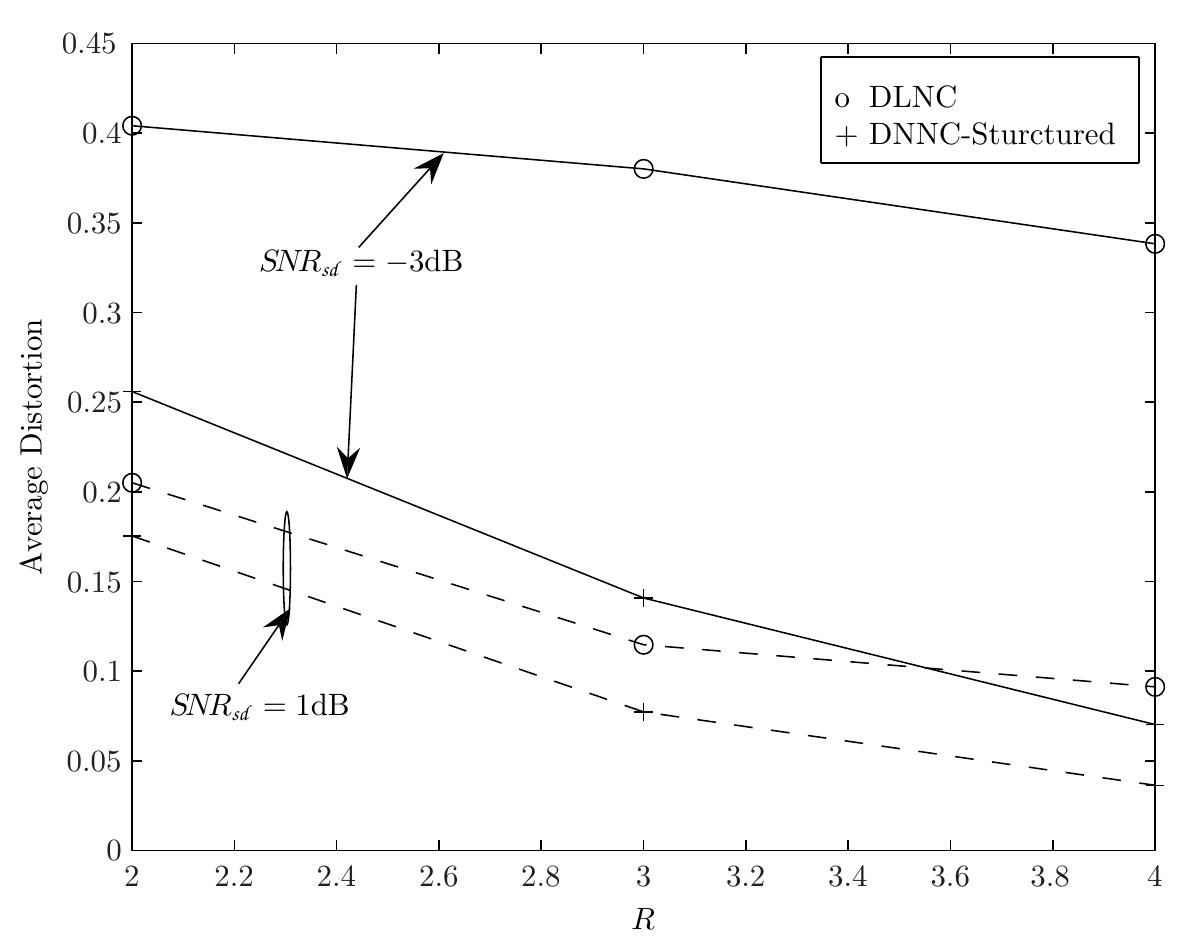}
\caption{Average distortion vs. rate for an OMARC with two sources,
$R_r=R$, $S\!N\!R_{\mathscr{s}\mathscr{r}}=\rm 10dB$ and
$S\!N\!R_{\mathscr{r}\mathscr{d}}=\rm 7dB$.}\label{F7}
\end{figure}

\begin{figure}[h!]
\centering
\includegraphics[width=4in]{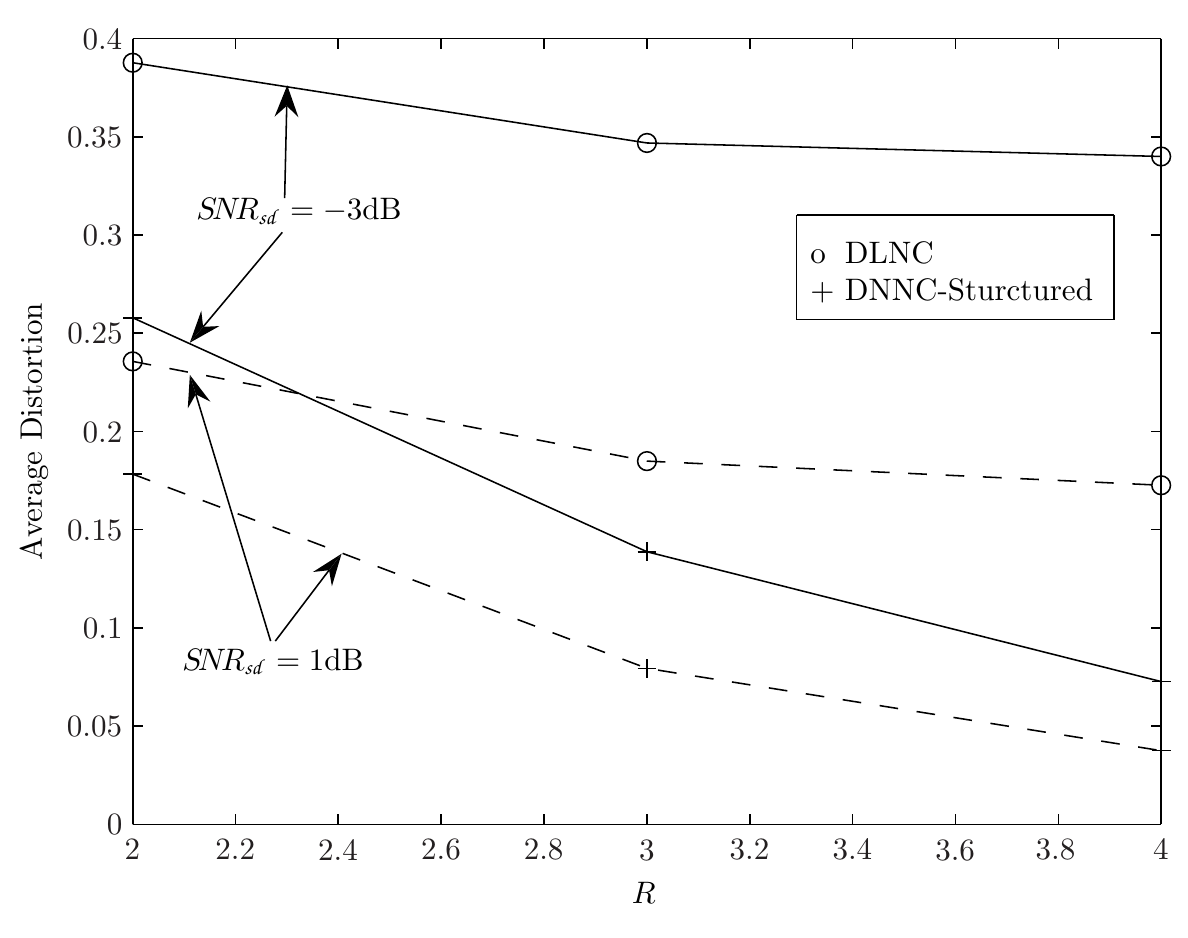}
\caption{Average distortion vs. rate for an OTN-BR-(2,3) with
$R_r=R$, $S\!N\!R_{\mathscr{s}\mathscr{r}}=\rm 10dB$ and
$S\!N\!R_{\mathscr{r}\mathscr{d}}=\rm 7dB$.}\label{F8}
\end{figure}

\begin{figure}[h!]
\centering
\includegraphics[width=4in]{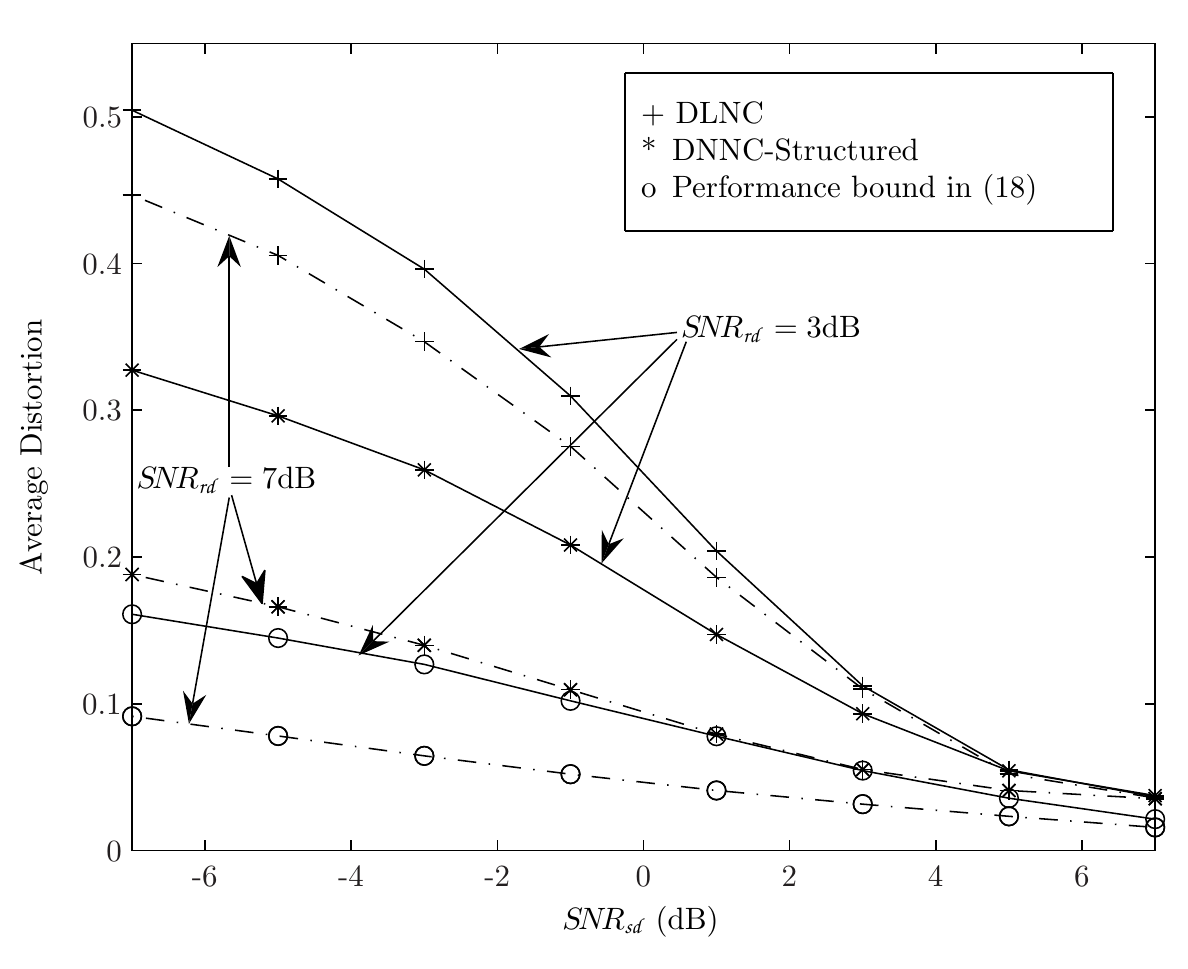}
\caption{Average distortion vs. source-destination channel SNR for
an OTN-BR-(2,3) with two sources, $R_r=R=3$ and
$S\!N\!R_{\mathscr{s}\mathscr{r}}=\rm 10dB$. Performance of linear
and proposed structured nonlinear network coding at the relay in
comparison with the performance bound in \eqref{E14}.}\label{F9}
\end{figure}

\end{document}